\documentclass [11pt ]{article}

\usepackage{pgf}

\usepackage [english ]{babel}
\usepackage [utf8 ]{inputenc}

\usepackage{amsmath,amssymb}
\usepackage{geometry}
\usepackage{array}
\usepackage{xspace}

\usepackage{amsmath,amsfonts,amssymb,amsopn,amsthm}
\theoremstyle{plain}

\newtheorem{lemma}{Lemma}
\newtheorem{proposition}{Proposition}
\newtheorem{corollary}{Corollary}
\newtheorem{theorem}{Theorem}
\theoremstyle{definition}
\newtheorem{example}{Example}
\newtheorem{definition}{Definition}

\theoremstyle{remark}
\newtheorem{remark}{Remark}

\newcommand{\F}{{\mathbb F}}

\newcommand{\Q}{{\mathbb Q}}

\newcommand{\Z}{{\mathbb Z}}
 \newcommand{\sC}{{\mathcal C}}
 \newcommand{\sN}{{\mathcal N}}

\newcommand{\be}{\begin{eqnarray}}
\newcommand{\ee}{\end{eqnarray}}
\newcommand{\nn}{{\nonumber}}

\newcommand{\Tr}{{\rm Tr}}

\makeatother

\newcommand{\wt}{wt}

\newcommand{\Wa} [1 ]{\widehat{\chi_{#1}}}
\newcommand{\Supp}{{\rm Supp}}
\newcommand{\supp}{{\rm supp}}

\begin{document} 

\title{Several classes of minimal linear codes with few weights from weakly regular  plateaued functions}

\author{Sihem Mesnager{\thanks{Department of Mathematics, University of Paris VIII, University of Paris XIII, CNRS, UMR 7539 LAGA and Telecom ParisTech, Paris, France.
Email: smesnager@univ-paris8.fr}} \and Ahmet S{\i}nak{\thanks{ 
Department of Mathematics and Computer Sciences at Necmettin Erbakan University, Institute of Applied Mathematics at Middle East Technical University, Turkey and LAGA, UMR 7539, CNRS at Universities
of Paris VIII and Paris XIII, France.
 Email: sinakahmet@gmail.com}}}

\date{}

\maketitle

\begin{abstract}
Minimal linear codes have significant applications in secret sharing schemes and secure
two-party computation.
There are several  methods to construct linear codes, one of which  is based on functions over finite fields.  Recently, many construction methods of linear codes   based on functions have been proposed in the literature.  In this paper, we generalize the recent construction methods given by Tang et al. in  [IEEE Transactions on Information Theory, 62(3), 1166-1176, 2016] to weakly regular plateaued functions over finite fields of odd characteristic. We first construct three weight  linear codes from weakly regular plateaued functions based on the second generic construction and   determine their weight distributions.
We next give a subcode with two or three weights of each constructed code as well as  its parameter.
We  finally   show that the constructed codes in this paper  are minimal,
 which  confirms  that 
 the secret sharing schemes based on their dual codes have the nice access structures.

 \end{abstract}

\section{Introduction}
Linear codes have diverse applications in secret sharing schemes,
authentication codes, communication, data storage devices, consumer electronics,  association schemes, strongly regular graphs and secure two-party computation. 
Indeed, as a special class of linear codes, minimal linear codes have significant applications in secret sharing schemes and secure two-party computation.
Constructing linear codes with perfect parameters, an interesting research topic in cryptography and  coding theory,
has been widely studied   in the literature. 
There are several  methods to construct linear codes, one of which  is based on functions over finite fields  (see \emph{e.g.} \cite{ding2016construction,ding2018minimal,ding2015class,mesnager2017linear,tang2016linear,zhou2016linear}).   
 Two generic constructions (say, \emph{first} and \emph{second}) of linear codes  from functions have been distinguished from the others in the  literature.
Recently,    several constructions of linear codes   based on the second generic construction have been proposed and many linear codes with perfect parameters have been constructed.   In fact, Ding   has come up with   an interesting survey \cite{ding2016construction} devoted to the construction of binary linear codes from Boolean functions  based on the second generic construction. Bent  functions (mostly, quadratic and weakly regular bent functions) have been extensively used  to construct linear codes with   few weights. It was shown  in a few papers (see \emph{e.g.} \cite{ding2015class,tang2016linear,zhou2016linear}) that bent functions  lead to the construction of interesting linear codes with few weights based on the second generic construction. Indeed,  Tang et al. (2016) have constructed in \cite{tang2016linear}   two or three weight linear codes from weakly regular  bent functions over finite fields of odd characteristic  based on the second generic construction. This inspires us to construct linear codes    from weakly regular  plateaued functions over finite fields of odd characteristic.   Within this framework,  to construct new  linear codes  with few weights, we aim  to make use of   weakly regular plateaued functions for the first time in the second generic construction.

The paper is structured as follows.  Section \ref{preliminaries}   sets   main notations and gives  background in  finite field theory and  coding theory. 
In Section \ref{SectionExponential},  we present some results related to weakly regular plateaued functions,  which will be needed to construct     linear codes.
In Section \ref{Constructions}, we construct two or three weight  linear codes by using some weakly regular plateaued functions over finite fields of odd characteristic based on the second generic construction. We also determine the weight distributions of the constructed  codes.
Finally, in Section \ref{SectionSSS}, we analyze the minimality of the constructed  codes  and hence  observe that all nonzero codewords of the constructed codes are minimal for almost all cases. 

\section{Preliminaries}\label{preliminaries}
In this section, after  setting basic notations, we mention the connection between linear codes and secret sharing schemes. We end this section by recording some properties of weakly regular plateaued functions. 
 \subsection{Basic notations}
Herein after, we fix the following notations unless otherwise stated.
\begin{itemize}
\item For any set $E$,  $\# E$ denotes the cardinality of $E$ and $E^\star=E\setminus \{0\}$,
\item   $\Z$  is the  ring of integers and  $\Q$  is  the field of rational numbers,
\item  $\vert z\vert$ denotes  the  magnitude of   $z\in\mathbb{C}$, where $\mathbb{C}$ is the field of complex numbers,
\item  $p$ is an odd prime and $q=p^n$ is an $n$-th  power of $p$ with $n$ being a positive integer,
\item  $\F_{q}$ is the finite field with $q$ elements and $\F_{q}^{\star}=\langle \zeta \rangle$ is a cyclic group  with generator $\zeta$,
\item 
The  trace of $\alpha\in\F_{q}$ over $\mathbb {F}_{p}$ is defined by $\Tr^n(\alpha)=\alpha+\alpha^{p}+\alpha^{p^{2}}+\cdots+\alpha^{p^{n-1}}$,
\item $\xi_p=e^{2\pi i/p}$ is the complex primitive $p$-th root of unity, where $i=\sqrt {-1}$ is the complex primitive $4$-th root of unity,
\item $SQ$ and $NSQ$  denote the set of all  squares and non-squares in  $\F_p^{\star}$, respectively,
\item $\eta$ and $\eta_0$  are the quadratic characters of $\F_q^{\star}$ and $\F_p^{\star}$,
\item $p^*=\eta_0(-1)p=(-1)^{\frac{p-1}{2}}p$. Notice that $p^n=\eta_0^n(-1) \sqrt{p^*}^{2n}$. 

\end{itemize}

\noindent \textbf{Cyclotomic field $\mathbb{Q}(\xi_p)$.}
   A cyclotomic field $\Q(\xi_p)$ is obtained from the field $\Q$  by adjoining $\xi_p$. 
The ring of integers in $\mathbb{Q}(\xi_p)$ is defined as $\mathcal {O}_{\mathbb{Q}(\xi_p)}:=\mathbb{Z}(\xi_p)$. An integral basis  of $\mathcal {O}_{\mathbb{Q}(\xi_p)}$ is the set  
$\{\xi_p^i :  1\leq i\leq p-1\}.$ 
The  field extension $\mathbb{Q}(\xi_p)/\mathbb{Q}$ is Galois of degree $p-1$, and  the Galois group 
$Gal (\mathbb{Q}(\xi_p)/\mathbb{Q})=\{\sigma_a : a \in \F_p^{\star}\},$
 where the automorphism $\sigma_a$ of $\mathbb{Q}(\xi_p)$ is defined by $\sigma_a(\xi_p)=\xi_p^a$. 
The cyclotomic field $\mathbb{Q}(\xi_p)$ has a unique quadratic subfield $\mathbb{Q}(\sqrt{p^*})$.  For $a\in  \F_p^{\star}$, we have  $\sigma_a(\sqrt {p^*})= {\eta_0}(a)\sqrt{p^*}$. Hence, the Galois group $Gal(\mathbb{Q}(\sqrt{p^*})/\mathbb{Q})=\{1, \sigma_{\gamma}\}$, where  $\gamma \in NSQ$. 
For $a \in \F_p^{\star}$ and $b \in \F_p$, we clearly have
  $\sigma_a(\xi_p^b)=\xi_p^{ab}$ 
and   $\sigma_a(\sqrt {p^*}^n)={\eta_0}^n(a)\sqrt{p^*}^n$, which will be needed to prove our subsequent results.
The reader is referred  to \cite{ireland2013classical} for further reading   on cyclotomic fields.

\subsection{Linear codes and  Secret sharing schemes} 
 A  linear code $\sC$  of length $n$ and dimension $k$ over $\mathbb {F}_{p}$ is a $k$-dimensional   linear subspace of  an $n$-dimensional vector space  $\F_p^n$, which can be viewed as  an extension field $\F_{p^n}$.  
 An  element of $\sC$ is said to be   \textit{codeword}. The Hamming weight of a vector $\bold a=(a_0,\ldots, a_{n-1})\in \mathbb {F}_{p}^n$,  denoted by $wt(\bold a)$, is the cardinality of  its support  defined as $$\supp(\bold  a)=\{0 \leq i\leq n-1: a_i\not=0 \}.$$
 The minimum Hamming distance of   $\sC$ is the minimum Hamming weight of its nonzero codewords. 
A linear code $\sC$  of length $n$ and dimension $k$ over $\mathbb {F}_{p}$ with minimum Hamming distance $d$ is denoted by  $\left[n,k,d\right]_{p}$. Note that $d$ detects the error correcting capability of $\sC$. 
 Let $A_w$ denote the number of codewords with Hamming weight $w$ in $\sC$ of  length $n$. Then, $(1,A_1, \ldots, A_n)$ is the \textit{ weight distribution} of $\sC$ and the polynomial $1+A_1y + \cdots + A_ny^n$ is called the \textit{ weight enumerator} of $\sC$. 
The code  $\sC$ is called a  \textit{$t$-weight code} if the number of nonzero $A_w$ in the weight distribution is  $t$. 

We now state the covering problem of a linear code  $\sC$.
We say that a codeword $\bold a$  of $\sC$  covers another codeword $\bold b$ of $\sC$ if   $ \supp(\bold a) $ contains   $\supp(\bold b)$.
A nonzero codeword $\bold a$  of   $\sC$  is said to be \textit{minimal} if $\bold a$  covers only the codeword $j\bold a$ for every $j\in\F_p$, but no other nonzero codewords of $\sC$.
A linear code $\sC$ is said to be \textit{minimal} if every nonzero codeword of $\sC$ is minimal.    Determining the minimality of a linear code  over finite fields has been an attractive research topic in coding theory.   The \textit{covering problem}  of  a linear code $\sC$ is to find all    minimal codewords of  $\sC$.
This  problem is rather difficult   for general linear codes, but easy for a few  special  linear codes. 

The following lemma says that all  nonzero codewords of $\sC$ are minimal if  the Hamming weights of the nonzero codewords of  $\sC$ are too  close to each other. The reader also notices that a necessary and sufficient  condition on minimal linear codes over finite fields has been  presented  in \cite[Theorem 11]{heng2018minimal}. Indeed, it is worth noting  that \cite{heng2018minimal}  provides an infinite
family of minimal linear codes violating the following sufficient condition.
\begin{lemma}(Ashikhmin-Barg) \cite{ashikhmin1998minimal} \label{Minimality}
All   nonzero codewords of a linear code $\sC$ over $\F_p$ are minimal if
$$
\frac{p-1}{p}<\frac{w_{\min}}{w_{\max}},
$$
where $w_{\min}$ and $w_{\max}$ denote the minimum and maximum  weights of nonzero codewords in $\sC$, respectively.
\end{lemma}

 In a secret sharing scheme, a set of participants who can recover the secret value $s$ from their shares is called \textit{an access set}. 
The set of all access sets is said to be  \textit{the access structure} of a secret sharing scheme.   An access set  is said to be  \textit{a minimal access set} if  any of its proper subsets cannot recover $s$ from their shares.
We take only an interest in the set of all minimal access sets, which is said to be  \textit{a nice access structure}. 

From \cite[Lemma 16]{carlet2005linear},   there is a one-to-one correspondence between the set of minimal access sets of the secret sharing scheme based on   $\sC $ and the set of minimal codewords of the dual code $\sC^\perp$. 
Hence, to find the minimal access sets of the secret sharing scheme based on  $\sC$, it is sufficient to find the minimal codewords of $\sC^\perp$ whose first coordinate is $1$. 

The following proposition   describes the access structure of the secret sharing scheme based on a dual code of a minimal linear code.
\begin{proposition}\cite{carlet2005linear,ding2003covering}\label{Structure}
Let   $\sC$ be  a minimal linear $[n,k,d]_p$ code over $\F_p$ with the generator matrix
$G=[\bold g_0,\bold g_1,\ldots,\bold g_{n-1}]$. We denote by $d^\perp$  the minimum Hamming distance  of its dual code $\sC^\perp$.
Then  in the secret sharing scheme based on $\sC^\perp$, the number of participants is $n-1$, and there exist $p^{k-1}$  minimal access sets. 
\begin{itemize}
\item If $d^\perp=2$, the access structure is given as follows.
  If $\bold g_i$, $1\leq i\leq n-1$, is a multiple of $\bold g_0$, then participant $P_i$ must be  in all minimal access sets. 
 If $\bold g_i$, $1\leq i\leq n-1$, is not a multiple of $\bold g_0$, then $P_i$ must be in $(p-1)p^{k-2}$ out of $p^{k-1}$ minimal access sets.
\item If $d^\perp\geq 3$, for any fixed $1\leq t\leq \min\{k-1,d^\perp-2\}$, every set of $t$ participants is involved in $(p-1)^tp^{k-(t+1)}$ out of $p^{k-1}$ minimal access sets.
\end{itemize}

\end{proposition}

In the following subsection, we introduce some results on weakly regular plateaued functions.
\subsection{Weakly regular plateaued functions}
Let $f :  \F_{q} \longrightarrow  \mathbb {F}_{p}$  be a $p$-ary function.
The Walsh transform of $f$ is given by:
$$\Wa {f}(\beta )=\sum_{x\in  \F_{q}} {\xi_p}^{{f(x)}-\Tr^n (\beta x)}, \;\; \beta\in  \F_{q}.$$
 A function  $f$ is said to be \textit{balanced} over $\F_p$ if 
$f$ takes every value of $\F_p$ the same number of $p^{n-1}$ times, in other words, $\Wa f(0)=0$; otherwise, $f$ is called \textit{unbalanced}.
Notice that  $f$ can be recovered from $\Wa {f}$ by the inverse Walsh transform:
 \begin{equation}\label{inversetransform}
 \xi_p^{f(x)}=\frac{1}{p^n} \sum_{\beta\in\F_{q}} \Wa {f}(\beta) \xi_p^{\Tr^n(\beta x)}.
\end{equation}
Bent functions,  introduced first in characteristic 2  by Rothaus \cite{rothaus1976bent} in 1976, are the functions whose Walsh coefficients satisfy $ \vert \Wa {f}(\beta)  \vert^2=p^{n}$. A bent function $f$ is called \emph {regular bent} if for every $\beta\in\F_{q}$, $p^{-\frac{n}2} \Wa {f}(\beta)=\xi_p^{f^{\star}(\beta)}$ for some $p$-ary function 
 $f^{\star}: \F_{q}\rightarrow \mathbb{F}_p$, and $f$ is called \emph {weakly regular bent} if there exist a complex number $u$ with $\vert u\vert=1$ and a $p$-ary function 
 $f^{\star}$ such that $up^{-\frac{n}2} \Wa {f}(\beta)=\xi_p^{f^{\star}(\beta)}$ for all $\beta\in\F_{q}$. Notice that $f^{\star}$ is also a weakly regular bent function. 
As an extension of bent functions, the notion of  plateaued functions was introduced   first in characteristic 2  by Zheng and Zhang   \cite{zheng1999plateaued} in 1999.
A function $f$ is  said to be  $p$-ary $s$-plateaued   if $|\widehat{\chi_f}(\beta)|^2\in\{0,p^{n+s}\}$ for every $\beta\in \F_{q}$, where $s$ is an integer with $0\leq s\leq n$. The Walsh support of  plateaued   $f$ is defined by  $\Supp(\widehat{\chi_f})=\{\beta\in  \F_{q}:  |\widehat{\chi_f}(\beta)|^2= p^{n+s}\}$.   The \textit{Parseval identity} is given by
 $\sum_{\beta \in \F_{q}} | \widehat{\chi_f}(\beta )|^{2}=p^{2n}.$
The absolute Walsh distribution of   plateaued functions  follows from the Parseval identity. 
\begin{lemma} \cite{mesnager2015results}\label{SupportLemma}
Let $f:\F_{q}\to\F_{p}$ be an $s$-plateaued function. Then for $\beta \in\F_{q}$, $| \widehat{\chi_f}(\beta)|^2$ takes $p^{n-s}$ times the value $p^{n+s}$ and $p^n-p^{n-s}$ times the value $0$.
\end{lemma} 
\begin{definition}\label{Weaklyregularity} \cite{DCC}
Let   $f:\F_{q}\to\F_{p}$ be a $p$-ary $s$-plateaued  function, where $s$  is an integer with $0 \leq s\leq n$. Then, 
$f$ is called \emph{weakly regular $p$-ary $s$-plateaued}  if  there exists a complex number $u$ having unit magnitude (in fact,   $|u|=1$ and $u$ does not depend on $\beta$) 
such that 
\be\nn\label{PlateauedWalshh}
\Wa {f}(\beta)\in \left\{ 0, up^{\frac{n+s}2}\xi_p^{g(\beta)}\right \}
\ee
 for  all $\beta\in \F_{q}$, where  $g$ is  a $p$-ary function  over $\F_{q}$ with $g(\beta)=0$ for all $\beta\in  \F_{q} \setminus  \Supp(\Wa f)$; otherwise,  $f$ is called \emph{non-weakly regular $p$-ary $s$-plateaued}. In particular, weakly regular    $f$ is said to be   \emph{regular} if $u=1$.
\end{definition}
\begin{lemma}\label{WalshFact} \cite{DCC}
Let  $f:\F_{q}\to\F_{p}$ be a weakly regular  $s$-plateaued function. 
For  all $\beta\in \Supp(\Wa f)$,  we have
$
 \Wa {f}(\beta)=\epsilon \sqrt{p^*}^{n+s} \xi_p^{g(\beta)},
$
 where $\epsilon=\pm 1$ is    the sign of  $ \Wa f$ and  $g$ is a $p$-ary function over  $\Supp(\Wa f)$. 
\end{lemma}
The following lemma can be  derived from  Lemma \ref{WalshFact}.
\begin{lemma}\label{Walshinverse}
 Let  $f:\F_{q}\to\F_{p}$ be a weakly regular  $s$-plateaued function.  Then for $x\in  \F_{q}$,
\be\nn
 \sum_{\beta\in \Supp(\Wa f)} \xi_p^{g(\beta)+\Tr^n(\beta x)}=\epsilon \eta_0^n(-1)\sqrt{p^*}^{n-s}\xi_p^{f(x)},
\ee
where  $\epsilon=\pm 1$ is    the sign of  $ \Wa f$ and  $g$ is a $p$-ary function over $\F_{q}$ with  $g(\beta)=0$ for all $\beta\in  \F_{q} \setminus  \Supp(\Wa f)$. 
\end{lemma}
\begin{proof} 
 By the inverse Walsh transform in (\ref{inversetransform}), we have 
\be\nn
\begin{array}{ll}
\xi_p^{f(x)}&=\displaystyle  \frac{1}{p^n} \sum_{\beta\in  \F_{q}} \Wa {f}(\beta) \xi_p^{\Tr^n(\beta x)}=\displaystyle  \frac{1}{p^n} \sum_{\beta\in \Supp(\Wa f)}\epsilon \sqrt{p^*}^{n+s} \xi_p^{g(\beta)}\xi_p^{\Tr^n(\beta x)},
\end{array}
\ee
where we used that $\Wa {f}(\beta)=0$    for all $\beta\in  \F_{q} \setminus  \Supp(\Wa f)$. Hence, we get
\be\nn
\begin{array}{ll}
\displaystyle \sum_{\beta\in \Supp(\Wa f)} \xi_p^{g(\beta)+\Tr^n(\beta x)}=  \frac{1}{\sqrt{p^*}^{n+s}}\epsilon  p^n\xi_p^{f(x)}= \epsilon \eta_0^n(-1) \sqrt{p^*}^{n-s}\xi_p^{f(x)},
\end{array}
\ee
where we used in the last equality  that  $p^n=\eta_0^n(-1) \sqrt{p^*}^{2n}$.
\end{proof}
The following lemma is a direct consequence of  Lemmas \ref{SupportLemma} and \ref{Walshinverse}.
\begin{lemma}\label{Walshg}
 Let  $f:\F_{q}\to\F_{p}$ be a weakly regular  $s$-plateaued function with $ \Wa {f}(\beta)=\epsilon \sqrt{p^*}^{n+s} \xi_p^{g(\beta)}$ for every $\beta\in\Supp (\Wa f)$,  where  $\epsilon=\pm 1$ is    the sign of  $ \Wa f$ and  $g$ is a $p$-ary function over $\F_{q}$ with  $g(\beta)=0$ for all $\beta\in  \F_{q} \setminus  \Supp(\Wa f)$.  
Then, we get $\Wa g(0)=p^n-p^{n-s}+ \epsilon \eta_0^n(-1)\sqrt{p^*}^{n-s}\xi_p^{f(0)}.$
\end{lemma}


\begin{remark}
Notice that Lemma \ref{Walshg} confirms that $g$  can not be balanced.
\end{remark}
We now define the subset  of the set of weakly regular plateaued functions, which is going to be used to construct linear codes.
 Let $f:\F_{q}\to\F_{p}$ be a weakly regular  $p$-ary $s$-plateaued unbalanced function, where $0\leq s\leq n$, and we denote by  $WRP$   the set of such functions satisfying the following two properties:
\begin{itemize}
\item [$P1)$] $f(0)=0$, and
\item [$P2)$]   there  exists a positive even integer $h$ with $\gcd(h-1,p-1)=1$ such that 
$f(ax)=a^hf(x)$  for any $a\in\F_p^{\star}$ and $x\in \F_q$.
\end{itemize}

\begin{lemma}\label{Walshsupport}
Let $f\in WRP$. Then for any $\beta\in\Supp(\Wa f)$ (resp., $\beta\in\F_q\setminus\Supp(\Wa f)$), we have
$z\beta\in\Supp(\Wa f)$ (resp., $z\beta\in\F_q\setminus\Supp(\Wa f)$) for every $z\in\F_p^{\star}$.
\end{lemma}
\begin{proof}
For every $z\in\F_p^{\star}$ and $\beta\in \F_q$,
we have 
\be\nn
\Wa {f}(z\beta )=\sum_{x\in  \F_{q}} {\xi_p}^{{f(x)}-\Tr^n (z\beta x)}=
\sum_{x\in  \F_{q}} {\xi_p}^{{f(z^kx)}-z\Tr^n (\beta z^kx)}=
\sum_{x\in  \F_{q}} {\xi_p}^{{z^lf(x)}-z^l\Tr^n (\beta x)},
\ee
where $k$ is a positive odd integer such that $k(h-1)\equiv 1  \pmod {p-1}$ and $l=k+1$, and is  
\be\nn
\sigma_{z^l}\left(\sum_{x\in  \F_{q}} {\xi_p}^{{f(x)}-\Tr^n (\beta x)}\right)=\sigma_{z^l}\left(\Wa {f}(\beta )\right)=
\left\{ \begin{array}{ll}
0,& \mbox{ if }  \beta\in\F_q\setminus\Supp(\Wa f),\\
\epsilon \sqrt{p^*}^{n+s} \xi_p^{z^lg(\beta)},& \mbox{ if }  \beta\in\Supp(\Wa f),
\end{array}\right.
\ee
where we used in the last equality that $\eta_0^{n+s}(z^l)=1$.
Hence, the proof is complete.
\end{proof}
 Lemma \ref{Walshsupport} implies that  there exists a subset $P_S$ of the Walsh support of  $f\in WRP$    such that  
$
\Supp(\Wa {f})=\F_p^{\star}P_S=\{z\beta :z\in\F_p^{\star} \mbox{ and } \beta\in  P_S\},
$
where for each pair of distinct elements $\beta_1,\beta_2\in P_S$ we have $\frac{\beta_1}{\beta_2}\notin\F_p^{\star}$.

We now give a brief introduction to the quadratic functions  (see \emph{e.g.} \cite{helleseth2006monomial}).
Recall that  any quadratic function from  $\F_{p^n}$ to $\F_p$ having no linear term can be represented by
\be\label{Quadratic}
Q(x)=\sum_{i=0}^{\lfloor n/2\rfloor} \Tr^n(a_ix^{p^i+1}),
\ee
where $\lfloor x \rfloor$ denotes the largest integer less than or equal to $x$ and $a_i\in\F_{p^n}$ for $0\leq i\leq \lfloor n/2\rfloor$. 
Let  $A$ be  a corresponding $n\times n$ symmetric matrix with $Q(x)=x^TAx$ defined in   \cite{helleseth2006monomial} and  $L$ be a corresponding linearized polynomial over   $\F_{p^n}$  defined as
\be\nn\label{Linearized}
L(z)=\sum_{i=0}^{l}(a_iz^{p^i} +a_i^{p^{n-i}} z^{p^{n-i}}).
\ee
 The set of linear structures of quadratic function $Q$ is the kernel of $L$, defined as 
\be\label{LinearizedKernel}
\ker_{\F_p}(L)=\{z\in \F_{p^n} : Q(z+y)=Q(z)+Q(y), \forall y\in \F_{p^n}\},
\ee
which is an $\F_p$-linear subspace of $ \F_{p^n}$.
 Let  the dimension of $\ker_{\F_p}(L)$  be  $s$ with $0\leq s\leq n$. Notice that by \cite[Proposition 2.1]{hou2004solution}, the rank of $A$ is equal to $n-s$. 
It was shown in  \cite{helleseth2006monomial} that a quadratic function $Q$ is bent if and only if $s=0$; equivalently,  $A$ is nonsingular, i.e.,  has full rank. Hence, 
we have the following natural consequence (see \emph{e.g.} \cite[Proposition 2]{helleseth2006monomial} and \cite[Example 1]{mesnager2015results}).

\begin{proposition} 
Any quadratic function $Q$ is $s$-plateaued if and only if the dimension of the kernel of $L$ defined in (\ref{LinearizedKernel}) is equal to $s$; equivalently, the rank of $A$ is $n-s$.  
\end{proposition}

Indeed,  from \cite[Proposition 1]{helleseth2006monomial} and \cite[Theorem 4.3]{ccesmelioglu2010construction} we  have also the following reasonable fact. 
\begin{proposition}
 The sign of Walsh transform of quadratic functions does not  depend on inputs which means that any quadratic function is weakly regular plateaued. Namely, there is no quadratic  non-weakly regular plateaued  functions.  
\end{proposition}

\begin{remark}
All  quadratic unbalanced functions defined in (\ref{Quadratic}) are in the set WRP. Hence, all of these functions can be used to construct linear codes in this paper. 
\end{remark}

\begin{example}\label{Example1}
 The function $f:\F_{3^4}\to \F_{3}$ defined as $f(x)=\Tr^4( \zeta x^{10} + \zeta^{51}x^4 + \zeta^{68} x^2)$,  where $\F_{3^4}^{\star}=\langle \zeta \rangle$ with $\zeta^4 + 2\zeta^3 + 2=0$,   is the quadratic $2$-plateaued unbalanced  function in the set WRP with
 $$\Wa {f}(\beta)\in\{0,\epsilon \eta^3_0(-1)3^3\xi_3^{g(\beta)}\}=\{0,-27,- 27\xi_3,- 27\xi_3^2\}$$
for all $\beta\in \F_{3^4}$, where $\epsilon=1$,   $\eta_0(-1)=-1$ and $g$ is an unbalanced $3$-ary function with $g(0)=0$, $g(\zeta^{2})=g(\zeta^{33})=g(\zeta^{42})=g(\zeta^{73})=1$ and $g(\zeta^{5})=g(\zeta^{6})=g(\zeta^{45})=g(\zeta^{46})=2$.  
Clearly, we have    $\Supp(\Wa {f})=\{0,\zeta^{2},\zeta^{5},\zeta^{6},\zeta^{33},\zeta^{42},\zeta^{45},\zeta^{46},\zeta^{73}\}$ and 
 $P_S=\{0,\zeta^{2}=2\zeta^{42},\zeta^{5}=2\zeta^{45},\zeta^{6}=2\zeta^{46},\zeta^{33}=2\zeta^{73}\}$.
\end{example}

We conclude  this section by recording the following necessary results. 
\begin{lemma}\cite{tang2016linear} \label{TangLemma}
Let $a_i,b_i\in\Z$  for every $i\in\F_p$ such that  
$\sum_{i=0}^{p-1}a_i\equiv\sum_{i=0}^{p-1}b_i \mod 2$ and 
$\sum_{i=0}^{p-1}a_i\xi_p^i\equiv\sum_{i=0}^{p-1}b_i \xi_p^i\mod 2$.
Then, $a_i \equiv b_i \mod 2$  for every $i\in\F_p$.
\end{lemma}

\begin{lemma}\cite{lidl1997finite} \label{Lemma3}
For $a\in\F_q^{\star}$, we have
$
\begin{array}{ll}  \sum_{x\in\F_q}\xi_p^{\Tr^n(ax^2)}=(-1)^{n-1}\eta(a)\sqrt{p^*}^{n}.
 \end{array}
$
In particular (in case of $n=1$),
for $a\in\F_p^{\star}$, we have
\be\nn
 \sum_{x\in\F_p}\xi_p^{ax^2}
=\left\{\begin{array}{ll}
\sqrt{p^*},  &  a\in SQ, \\
-\sqrt{p^*}, & a\in NSQ.
 \end{array}\right.
\ee
\end{lemma}

\begin{proposition}\label{Proposition4}
Let $f\in WRP$, then $g(0)=0.$
\end{proposition}
\begin{proof}
From Lemmas \ref{TangLemma} and \ref{Lemma3},  the proof can be completed by using the same argument used in the proof of \cite[Proposition 4]{tang2016linear}.
\end{proof}

One can immediately observe the following proposition from the proof of Lemma \ref{Walshsupport}.
\begin{proposition}\label{Proposition5}
Let $f\in WRP$, then there exists a  positive even integer $l$ with $\gcd(l-1,p-1)=1$ such that  $g(a \beta)=a^lg(\beta)$  for any $a\in\F_p^{\star}$ and $\beta \in \Supp(\Wa f)$.
\end{proposition}

\begin{lemma}\label{Lemma6} \cite{lidl1997finite} We have 
\begin{itemize}
\item  [$i.)$ ]  $  \sum_{a \in\mathbb{F}^{\star}_p} \eta_0(a)=0$,
\item  [$ii.)$ ]  $\sum_{a\in\mathbb{F}^{\star}_p} \eta_0(a) \xi_p^{a}=\sqrt {p^*}
=\left\{\begin{array}{ll}\sqrt{p},  & \mbox{ if }  p \equiv 1\pmod 4,~ \\
i\sqrt{p},& \mbox{ if } p\equiv 3\pmod 4,
  \end{array}\right.$
\item  [$iii.)$ ]  $\sum_{a\in \mathbb{F}^{\star}_p}\xi_p^{ a  b}=-1$ for any $b\in \mathbb{F}^{\star}_p$.
\end{itemize}
\end{lemma}

 \section{Exponential sums from weakly regular plateaued functions}\label{SectionExponential}
In this section, we present some results on exponential sums related to   weakly regular plateaued functions,   which are going to be needed in Section \ref{Constructions} to construct linear codes.
Actually, these results were given in    \cite{tang2016linear} for weakly regular bent functions; however, for the sake of completeness, we state them translated into our framework and  give their proofs.

We begin this section with the following lemma, which will be used to find the length of a linear  code.
\begin{lemma}\label{Lemma7}
Let $f:\F_q\to\F_p$ be an unbalanced function with $\Wa f (0)=\epsilon \sqrt{p^*}^{n+s}$, where $\epsilon=\pm 1$  is    the sign of  $ \Wa f$. For   $j\in\F_p$, define $\sN_f(j)=\#\{x\in\F_q : f(x)=j\}.
$ Then, if $n+s$ is even,  
\be\nn
\sN_f(j)=\left\{\begin{array}{ll}
p^{n-1}+\epsilon \eta_0 (-1)(p-1)\sqrt{p^*}^{n+s-2},  & \mbox{ if } j=0, \\
p^{n-1}-\epsilon \eta_0 (-1)\sqrt{p^*}^{n+s-2},& \mbox{ if } j\in\F_p^{\star},
 \end{array}\right.
\ee
 if $n+s$ is odd,  
$
\sN_f(j)=\left\{\begin{array}{ll}
p^{n-1},  & \mbox{ if } j=0, \\
p^{n-1}+\epsilon \sqrt{p^*}^{n+s-1}, 
& \mbox{ if }  j \in SQ,\\
p^{n-1}-\epsilon  \sqrt{p^*}^{n+s-1},& \mbox{ if }  j\in NSQ.
 \end{array}\right.
$\end{lemma}
\begin{proof}
Clearly, we have
$
\sum_{j=0}^{p-1}\sN_{f}(j)\xi^j_p=\epsilon\sqrt{p^*}^{n+s};
$
namely,
$
\sum_{j=0}^{p-1}\sN_{f}(j)\xi^j_p-\epsilon\sqrt{p^*}^{n+s}=0.
$  If $n+s$ is even,  then for all $ j\in\F_p^{\star}$ we have 
$\sN_{f}(j)=a$ and $\sN_{f}(0)=a+\epsilon \sqrt{p^*}^{n+s}
$ for some constant integer $a$ since $\sum_{j=0}^{p-1}x^j$  is the minimal polynomial of $\xi_p$ over $\Q$. 
 Hence, since $\sum_{j=0}^{p-1} \sN_{f}(j)=p^{n}$, we get 
$a+\epsilon \sqrt{p^*}^{n+s}+(p-1)a=p^n$
from which we deduce that 
$
a=p^{n-1}-\epsilon \eta_0^{(n+s)/2}(-1)  p^{(n+s-2)/2}.
$
  If $n+s$ is odd, then we get
\be\nn
\sum_{j=0}^{p-1}\sN_{f}(j)\xi^j_p-\epsilon  \sqrt{p^*}^{n+s-1}\sum_{j=1}^{p-1} \eta_0(j)\xi^j_p=0,
\ee 
where we used  the fact that  $\sum_{j=1}^{p-1} \eta_0(j) \xi_p^{j}=\sqrt {p^*}$
 by Lemma \ref{Lemma6} $(ii)$, equivalently,
\be\nn
\sN_{f}(0)+ \sum_{j=1}^{p-1}\left(\sN_{f}(j)-\epsilon  \eta_0(j)\sqrt{p^*}^{n+s-1} \right)\xi^j_p=0.
\ee
As in the even case,    for all $ j\in\F_p^{\star}$ we have
$
\sN_{f}(j)=\sN_{f}(0)+\epsilon\eta_0(j)\sqrt{p^*}^{n+s-1}
$   
and hence
$
p\sN_{f}(0)+\epsilon  \sqrt{p^*}^{n+s-1}\sum_{j=1}^{p-1} \eta_0(j)=p^{n}.
$ Then we get $\sN_{f}(0)=p^{n-1}$ by Lemma \ref{Lemma6} $(i)$. Thus,
the proof is ended.
\end{proof}

 The following lemma has a crucial role in determining the weight distributions of a  linear code.  
\begin{lemma}\label{Lemma8}
Let $f\in WRP$ with  $ \Wa {f}(\beta)=\epsilon \sqrt{p^*}^{n+s} \xi_p^{g(\beta)}$ for every  $\beta\in \Supp(\Wa f)$. 
For $j\in\F_p$, define 
$\sN_g(j)=\#\{\beta\in\Supp(\Wa f) : g(\beta)=j\}.$ Then, if $n-s$ is even,
\be\nn
\sN_g(j)=\left\{\begin{array}{ll}
p^{n-s-1} + \epsilon \eta_0^{n+1}(-1)(p-1) \sqrt{p^*}^{n-s-2},  & \mbox{ if } j=0, \\
p^{n-s-1} -\epsilon \eta_0^{n+1}(-1) \sqrt{p^*}^{n-s-2},& \mbox{ if } j\in\F_p^{\star},
 \end{array}\right.
\ee
  if $n-s$ is odd,
$
\sN_g(j)=\left\{\begin{array}{ll}
p^{n-s-1},  & \mbox{ if } j=0, \\
p^{n-s-1} +\epsilon  \eta_0^{n}(-1) \sqrt{p^*}^{n-s-1},& \mbox{ if } j \in SQ,\\
p^{n-s-1} - \epsilon\eta_0^n(-1)   \sqrt{p^*}^{n-s-1}, & \mbox{ if }  j\in NSQ.
 \end{array}\right.
$\end{lemma}
\begin{proof} 
By Lemma \ref{Walshinverse}, for $x=0$, we have
$\displaystyle\sum_{\beta\in \Supp(\Wa f)} \xi_p^{g(\beta)}=\epsilon \eta_0^n(-1)\sqrt{p^*}^{n-s},$
equivalently,
\be\nn
\sum_{j=0}^{p-1}\sN_{g}(j)\xi^j_p-\epsilon \eta_0^n(-1)\sqrt{p^*}^{n-s}=0.
\ee
  If $n-s$ is even,  then for all $ j\in\F_p^{\star}$ we have 
$\sN_{g}(j)=a$ and $\sN_{g}(0)=a+\epsilon \eta_0^n(-1)\sqrt{p^*}^{n-s}
$ for some constant integer $a$. 
 Hence,  since $\#\Supp (\Wa f)=p^{n-s}$, we get 
\be\nn
\sum_{j=0}^{p-1} \sN_{g}(j)=a+\epsilon \eta_0^n(-1)\sqrt{p^*}^{n-s}+(p-1)a=p^{n-s}.
\ee 
  If $n-s$ is odd, with the same argument used in the proof of Lemma \ref{Lemma7}, 
we get
$\sN_{g}(j)=\sN_{g}(0)+\epsilon \eta_0^n(-1)\eta_0(j)\sqrt{p^*}^{n-s-1}$    for all $ j\in\F_p^{\star}$ and hence  
\be\nn
p\sN_{g}(0)+\epsilon \eta_0^n(-1) \sqrt{p^*}^{n-s-1}\sum_{j=1}^{p-1} \eta_0(j)=p^{n-s}.
\ee
 Then we conclude $\sN_{g}(0)=p^{n-s-1}$ from Lemma \ref{Lemma6} $(i)$, which completes the proof.
\end{proof}

\begin{lemma}\label{Lemma9}
Let $f:\F_q\to\F_p$ be an unbalanced function with $\Wa f (0)=\epsilon \sqrt{p^*}^{n+s}$, where $\epsilon=\pm 1$  is    the sign of  $ \Wa f$. Then,
\be\nn
\begin{array}{ll}
\displaystyle \sum_{y\in\F_p^{\star}} \sum_{x\in\F_q}\xi_p^{yf(x)}=\left\{\begin{array}{ll}
 \epsilon(p-1) \sqrt{p^*}^{n+s},  & \mbox{ if  } n+s \mbox{ is even}, \\
0, & \mbox{ if } n+s  \mbox{ is odd}.
 \end{array}\right.
 \end{array}
\ee
\end{lemma}
\begin{proof} From the definition, one can immediately observe that
\be\nn
\begin{array}{ll}
\displaystyle \sum_{y\in\F_p^{\star}} \sum_{x\in\F_q}\xi_p^{yf(x)}
=\sum_{y\in\F_p^{\star}} \sigma_y\left(\sum_{x\in\F_q}\xi_p^{f(x)}\right)
=\sum_{y\in\F_p^{\star}} \sigma_y\left( \epsilon \sqrt{p^*}^{n+s}\right)
=\epsilon \sqrt{p^*}^{n+s}\sum_{y\in\F_p^{\star}}\eta^{n+s}_0(y).
 \end{array}
\ee
Hence, the assertion   follows clearly from Lemma \ref{Lemma6} $(i)$.
\end{proof}

\begin{lemma}\label{Lemma10}
Let $f\in WRP$. For   $\beta\in\F_q^{\star}$, define $
A=\sum_{y,z\in\F_p^{\star}} \sum_{x\in\F_q}\xi_p^{yf(x)-z\Tr^n(\beta x)}.
$ Then,  for every $\beta\in\F_q^{\star}\setminus \Supp(\Wa f),$   we have $A=0$, and for every  $0\neq \beta\in \Supp(\Wa f),$  if $n+s$ is even, then
\be\nn
A=\left\{\begin{array}{ll}
  \epsilon  (p-1)^2 \sqrt{p^*}^{n+s},  & \mbox{ if } g(\beta)=0, \\
- \epsilon  (p-1) \sqrt{p^*}^{n+s},& \mbox{ if } g(\beta)\neq 0,
 \end{array}\right.
\ee
 if $n+s$ is odd, then
$
A=\left\{\begin{array}{ll}
0,  & \mbox{ if } g(\beta)=0, \\
 \epsilon  \eta_0(g(\beta)) (p-1) \sqrt{p^*}^{n+s+1},& \mbox{ if } g(\beta)\neq 0.
 \end{array}\right.
$
\end{lemma}

\begin{proof} By Lemma \ref{Walshsupport}, 
  for every  $\beta\in \Supp(\Wa f),$ we have $\Wa f(z\beta)=\epsilon \sqrt{p^*}^{n+s}\xi_p^{g(z\beta)}$ for any $z\in\F_p^{\star}$.
Then for every  $0\neq \beta\in \Supp(\Wa f),$  we have 
\be\nn
\begin{array}{ll}
A=\displaystyle\sum_{y,z\in\F_p^{\star}} \sigma_y\left(\sum_{x\in\F_q}\xi_p^{f(x)-\Tr^n(z\beta x)}\right)&=\displaystyle\sum_{y,z\in\F_p^{\star}} \sigma_y(\Wa f(z\beta))\\
&=\displaystyle\sum_{y,z\in\F_p^{\star}} \sigma_y(\epsilon\sqrt{p^*}^{n+s}\xi_p^{g(z\beta)})\\
&=\displaystyle\sum_{y,z\in\F_p^{\star}} \sigma_y(\epsilon \sqrt{p^*}^{n+s}\xi_p^{z^lg(\beta)})\\
&=\epsilon (p-1) \sqrt{p^*}^{n+s}\displaystyle\sum_{y\in\F_p^{\star}}\eta_0^{n+s}(y)\xi_p^{yg(\beta)}),
 \end{array}
\ee
where we used Proposition \ref{Proposition5} in the fourth equality and used the fact that $z^l$ is a square in $\F_p^{\star}$ for any $z\in\F_p^{\star}$ in the last equality.
When $n+s$ is even, 
$A=  \epsilon (p-1) \sqrt{p^*}^{n+s}\sum_{y\in\F_p^{\star}}\xi_p^{yg(\beta)},$
which is $ \epsilon (p-1)^2 \sqrt{p^*}^{n+s}$ if $g(\beta)=0$; otherwise,
 $-\epsilon (p-1) \sqrt{p^*}^{n+s}$ from Lemma \ref{Lemma6} ($iii$).
When $n+s$ is odd,  we clearly have $A=0$ if $g(\beta)=0$; otherwise,   conclude  that 
\be\nn
A=  \epsilon (p-1)\sqrt{p^*}^{n+s}\sigma_{g(\beta)}\left(\sum_{y\in\F_p^{\star}}\eta_0(y)\xi_p^{y}\right)= \epsilon\eta_0(g(\beta)) (p-1) \sqrt{p^*}^{n+s+1},
\ee
where we used Lemma \ref{Lemma6} ($ii$).
For every  $ \beta\in\F_q^{\star}\setminus \Supp(\Wa f)$,  we immediately get $A=0$. Hence, the proof is complete.
\end{proof}
 The following lemma has   a significant role in finding the Hamming weights of the codewords of a linear code.
\begin{lemma}\label{Lemma11}
Let $f\in WRP$ with $\Wa f (0)=\epsilon \sqrt{p^*}^{n+s}$. For $\beta\in\F_q^{\star}$, define 
$\sN_{f,\beta}=\#\{x\in\F_q : f(x)=0 \mbox{ and } \Tr^n(\beta x)=0\}.$
Then,  for every $\beta\in\F_q^{\star}\setminus \Supp(\Wa f),$   
\be\nn
\sN_{f,\beta}=
\left\{\begin{array}{ll}
    p^{n-2}+\epsilon (p-1) \sqrt{p^*}^{n+s-4},  & \mbox{ if } n+s \mbox{ is even}, \\
  p^{n-2},& \mbox{ if } n+s \mbox{ is odd},
 \end{array}\right.
\ee
 and for every  $0\neq \beta\in \Supp(\Wa f),$ 
  if $n+s$ is even, then
\be\nn
\sN_{f,\beta}=\left\{\begin{array}{ll}
p^{n-2} + \epsilon \eta_0 (-1)(p-1) \sqrt{p^*}^{n+s-2},  & \mbox{ if } g(\beta)=0, \\
p^{n-2},& \mbox{ if } g(\beta)\neq 0,
 \end{array}\right.
\ee
  if $n+s$ is odd, then
$
\sN_{f,\beta}=\left\{\begin{array}{ll}
p^{n-2},  & \mbox{ if } g(\beta)=0, \\
p^{n-2}+\epsilon   (p-1)  \sqrt{p^*}^{n+s-3},& \mbox{ if } g(\beta)\in SQ,\\
p^{n-2}-\epsilon   (p-1)  \sqrt{p^*}^{n+s-3},& \mbox{ if } g(\beta)\in NSQ.
 \end{array}\right.
$
\end{lemma}
\begin{proof} By the definition of $\sN_{f,\beta}$, one can  observe that
\be\nn
\begin{array}{ll}
\sN_{f,\beta}&=p^{-2}\displaystyle\sum_{x\in\F_q}\left(\sum_{y\in\F_p} \xi_p^{yf(x)}\right)\left(\sum_{z\in\F_p} \xi_p^{-z\Tr^n(\beta x)}\right)\\
&=p^{n-2}+p^{-2}\displaystyle \sum_{y\in\F_p^{\star}}\sum_{x\in\F_q} \xi_p^{yf(x)}
+p^{-2}\displaystyle \sum_{y,z\in\F_p^{\star}}\sum_{x\in\F_q}  \xi_p^{yf(x)-z\Tr^n(\beta x)}.\\
\end{array}
\ee
Hence, the proof is concluded from Lemmas \ref{Lemma9} and \ref{Lemma10}.
\end{proof}
The following can be immediately observed as in Lemma \ref{Lemma9}.
\begin{lemma}\label{Lemma12}
Let $f:\F_q\to\F_p$ be an unbalanced function with $\Wa f (0)=\epsilon \sqrt{p^*}^{n+s}$, where $\epsilon=\pm 1$ is the sign of $\Wa{f}$. Then, we have
\be\nn
\displaystyle \sum_{y\in\F_p^{\star}} \sum_{x\in\F_q}\xi_p^{y^2f(x)}=
 \epsilon(p-1) \sqrt{p^*}^{n+s}.
\ee
\end{lemma}

\begin{lemma}\label{Lemma13}
Let $f\in WRP$. For $\beta\in\F_q^{\star}$, define $
A=\sum_{y,z\in\F_p^{\star}} \sum_{x\in\F_q}\xi_p^{y^2f(x)-z\Tr^n(\beta x)}.
$ Then,  for every $\beta\in\F_q^{\star}\setminus \Supp(\Wa f),$   we have $A=0$, and for every  $ 0\neq \beta\in \Supp(\Wa f),$ we have
\be\nn
A=\left\{\begin{array}{ll}
  \epsilon  (p-1)^2 \sqrt{p^*}^{n+s},  & \mbox{ if } g(\beta)=0, \\
 \epsilon  (p-1) \sqrt{p^*}^{n+s}(\sqrt{p^*}-1),& \mbox{ if } g(\beta)\in SQ,\\
- \epsilon  (p-1) \sqrt{p^*}^{n+s}(\sqrt{p^*}+1),& \mbox{ if }  g(\beta)\in NSQ.
 \end{array}\right.
\ee

\end{lemma}

\begin{proof} 
As in the proof of Lemma \ref{Lemma10}, we get 
$A= \sum_{y,z\in\F_p^{\star}} \sigma_{y^2}( \Wa f(z\beta))$ for every $\beta\in\F_q^{\star}$.
For every  $ \beta\in\F_q^{\star}\setminus \Supp(\Wa f),$  we clearly have $A=0$.
For every  $ 0\neq \beta\in \Supp(\Wa f),$ 
we get
\be\nn
\begin{array}{ll}
A=\displaystyle\sum_{y,z\in\F_p^{\star}} \sigma_{y^2}( \epsilon \sqrt{p^*}^{n+s}\xi_p^{g(z\beta)})&=\displaystyle \sum_{y,z\in\F_p^{\star}} \sigma_{y^2}(\epsilon \sqrt{p^*}^{n+s} \xi_p^{z^lg(\beta)})\\
&=\displaystyle \epsilon \sqrt{p^*}^{n+s}\sum_{y,z\in\F_p^{\star}}\sigma_{y^2}(\xi_p^{z^lg(\beta)})\\
&=\displaystyle \epsilon \sqrt{p^*}^{n+s}(p-1)\sum_{y\in\F_p^{\star}}\xi_p^{y^2g(\beta)}\\
&=\displaystyle \epsilon \sqrt{p^*}^{n+s}(p-1)\sum_{y\in\F_p}\xi_p^{y^2g(\beta)}-\epsilon \sqrt{p^*}^{n+s}(p-1),
\end{array}
\ee
where  in the second equality we used Proposition  \ref{Proposition5} and in the fourth equality we used the fact that $y^2z^l$ runs through all squares in $\F_p^{\star}$ when $y$ ranges  over $\F_p^{\star}$ for any fixed $z\in \F_p^{\star}$. Hence the assertion follows  from Lemma \ref{Lemma3}.
\end{proof}

 The following lemma   has a significant role in finding the Hamming weights of the codewords of a linear code. 

\begin{lemma}\label{Lemma14}
Let $f\in WRP$. For $\beta\in\F_q^{\star}$, define
\be\nn
\begin{array}{ll}
\sN_{sq,\beta}&=\#\{x\in\F_q : f(x)\in SQ \mbox{ and } \Tr^n(\beta x)=0\},\\
\sN_{nsq,\beta}&=\#\{x\in\F_q : f(x)\in NSQ \mbox{ and } \Tr^n(\beta x)=0\}.
\end{array}
\ee
Then,  for every $\beta\in\F_q^{\star}\setminus \Supp(\Wa f),$   if $n+s$ is even,   
$\sN_{sq,\beta}=\sN_{nsq,\beta} = \frac{p-1}{2}\left(p^{n-2} -\epsilon \sqrt{p^*}^{n+s-4}\right)$,
 if $n+s$ is odd,
\be\nn
\begin{array}{ll}
\sN_{sq,\beta}&=\frac{p-1}{2}\left(p^{n-2} + \epsilon \eta_0 (-1) \sqrt{p^*}^{n+s-3}\right),\\
\sN_{nsq,\beta}&= \frac{p-1}{2}\left(p^{n-2} - \epsilon \eta_0 (-1)  \sqrt{p^*}^{n+s-3}\right).
 \end{array}
\ee
For every  $0\neq \beta\in \Supp(\Wa f),$ 
if $n+s$ is even, then 
\be\nn\begin{array}{ll}
\sN_{sq,\beta}&=\left\{\begin{array}{ll}
\frac{p-1}{2}\left(p^{n-2} -\epsilon \eta_0 (-1) \sqrt{p^*}^{n+s-2}\right),  & \mbox{ if } g(\beta)=0 \mbox{ or } g(\beta) \in NSQ, \\
\frac{p-1}{2}\left(p^{n-2} + \epsilon\eta_0 (-1) \sqrt{p^*}^{n+s-2}\right),& \mbox{ if } g(\beta) \in SQ,
 \end{array}\right.\\
\sN_{nsq,\beta}&=\left\{\begin{array}{ll}
\frac{p-1}{2}\left(p^{n-2} -\epsilon \eta_0 (-1) \sqrt{p^*}^{n+s-2}\right),  & \mbox{ if } g(\beta)=0 \mbox{ or } g(\beta) \in SQ, \\
\frac{p-1}{2}\left(p^{n-2} + \epsilon \eta_0 (-1) \sqrt{p^*}^{n+s-2}\right),& \mbox{ if } g(\beta) \in NSQ,
 \end{array}\right.
\end{array}
\ee
if $n+s$ is odd, then
\be\nn\begin{array}{ll}
\sN_{sq,\beta}&=\left\{\begin{array}{ll}
\frac{p-1}{2}\left(p^{n-2}+\epsilon \sqrt{p^*}^{n+s-1}\right),  & \mbox{ if } g(\beta)=0, \\
\frac{p-1}{2}\left(p^{n-2} - \epsilon  \sqrt{p^*}^{n+s-3}\right),& \mbox{ if } g(\beta) \in SQ,\\
\frac{p-1}{2}\left(p^{n-2} +\epsilon  \sqrt{p^*}^{n+s-3}\right),  & \mbox{ if }  g(\beta) \in NSQ, \\
 \end{array}\right.\\

\sN_{nsq,\beta}&=\left\{\begin{array}{ll}
\frac{p-1}{2}\left(p^{n-2} -\epsilon \sqrt{p^*}^{n+s-1}\right),  & \mbox{ if } g(\beta)=0, \\
\frac{p-1}{2}\left(p^{n-2} -\epsilon  \sqrt{p^*}^{n+s-3}\right),  & \mbox{ if } g(\beta) \in SQ, \\
\frac{p-1}{2}\left(p^{n-2} + \epsilon  \sqrt{p^*}^{n+s-3}\right), 
& \mbox{ if } g(\beta) \in NSQ.
 \end{array}\right.
\end{array}
\ee

\end{lemma}
\begin{proof}
In view of Lemmas \ref{Lemma11}, \ref{Lemma12} and  \ref{Lemma13},  the proof can be constructed with the same argument used in the proof of \cite[Lemma 14]{tang2016linear}.
\end{proof}

\section{Linear codes with few weights from weakly regular plateaued functions}\label{Constructions}
In this  section,  we generalize the recent construction methods of linear codes proposed by  Ding et al. \cite{ding2015linear,ding2015class} and  Tang et al.  \cite{tang2016linear} to weakly regular plateaued functions,  based on the second generic construction. We also record a subcode of any constructed  code.

\subsection{Two or three weight linear codes with their weight distributions}
In this subsection, to construct new linear codes, we make use of weakly regular plateaued functions in the construction method of linear codes proposed   by   Ding et al. \cite{ding2015class}. 
Let $f$ be a $p$-ary function from $\F_q$ to $\F_p$. Define a set 
\be\label{DefiningSet}
D_f=\{ x\in\F_q^{\star}: f(x)=0\}.
\ee
Assume $m=\# D_f$ and $D_f=\{d_1, d_2, \ldots, d_m\}$.
 The second generic construction of a linear code from   $f$ is obtained from  $D_f$   and a linear code involving $D_f$ is defined by
\be\label{LinearCodes}
\sC_{D_f}=\{c_\beta= (\Tr{^n}(\beta d_1), \Tr{^n}(\beta d_2), \ldots, \Tr{^n}(\beta d_m)) :  \beta \in\mathbb{F}_q\}.
\ee
 The set $D_f$ is usually called the {\em defining set} of the code $\sC_{D_f}$. The   code $\sC_{D_f}$ has length $m$ and  dimension at most $n$. This construction is generic in the sense that many classes of known codes could be produced
by selecting the defining set $D_f\subseteq \mathbb{F}_q$. For a general function $f$, determining the weight distribution of $\sC_{D_f}$ is little hard,  but  easy for some special functions. For example, the weight distribution of $\sC_{D_f}$ was determined by Ding et al. \cite{ding2015class}  for a quadratic function $f(x)=\Tr^n(x^2)$,  by Zhou et al. \cite{zhou2016linear} for  quadratic bent functions and by Tang et al. \cite{tang2016linear} for some weakly regular bent functions. We now solve this problem for some weakly regular plateaued functions. 

The Hamming weights of the codewords of the   code $\sC_{D_f}$  as well as  its length are  derived from  Lemmas \ref{Lemma7} and \ref{Lemma11}, and its weight distribution is determined by Lemmas \ref{SupportLemma} and \ref{Lemma8}.

\begin{theorem}\label{Theorem1}
Let  $n+s$ be an even integer and   $f\in WRP$. Then $\sC_{D_f}$ is the three-weight linear code with parameters 
$\left[ p^{n-1}-1+\epsilon \eta_0(-1)(p-1) \sqrt{p^*}^{n+s-2},n\right]_p$, where  $\epsilon=\pm 1$ is the sign of $\Wa{f}$. The Hamming weights of the codewords and the weight distribution of $\sC_{D_f}$  are as in  Table \ref{table1}.
\begin{table}[!htp]
\begin{center}	
\begin{tabular}{|c|c|c|}
\hline
Hamming weight $w$   & Multiplicity $A_w$ \\
\hline
\hline
\footnotesize{$0$} & \footnotesize{$1$}\\
\hline
\footnotesize{$(p-1)\left(p^{n-2}+\epsilon (p-1)\sqrt{p^*}^{n+s-4}\right)$} & \footnotesize{$p^n-p^{n-s}$}\\
\hline
\footnotesize{$ (p-1)p^{n-2}$} & \footnotesize{$p^{n-s-1} + \epsilon \eta_0^{n+1}(-1) (p-1) \sqrt{p^*}^{n-s-2} -1$}\\
\hline
\footnotesize{$(p-1)\left(p^{n-2} + \epsilon \eta_0(-1)\sqrt{p^*}^{n+s-2}\right)$} & \footnotesize{$(p-1)\left(p^{n-s-1} -\epsilon \eta_0^{n+1}(-1) \sqrt{p^*}^{n-s-2}\right)$}\\
\hline
\end{tabular}\end{center}	
\caption{
\label{table1}  The weight distribution of  $\sC_{D_f}$ when $n+s$ is even}
\end{table}
\end{theorem}

\begin{proof}
Clearly,  we get $\# D_f=\sN_f(0)-1= p^{n-1}-1+\epsilon \eta_0(-1)(p-1) \sqrt{p^*}^{n+s-2}$ by Lemma \ref{Lemma7} and  $\wt(c_\beta)=\# D_f-\sN_{f,\beta}+1$ for every $\beta \in\F_q^{\star}$ by Lemma \ref{Lemma11}.  Then,  for every $\beta\in\F_q^{\star}\setminus \Supp(\Wa f),$  we have
$\wt(c_\beta)=(p-1)(p^{n-2}+\epsilon (p-1)\sqrt{p^*}^{n+s-4}),
$
and  the number of such codewords $c_\beta$ is  $p^n-p^{n-s}$ by Lemma \ref{SupportLemma}.
 For every  $0\neq  \beta\in \Supp(\Wa f),$    we obtain
\be\nn
\wt(c_\beta)=
\left\{\begin{array}{ll}
 (p-1)p^{n-2} ,  & \mbox{ if } g(\beta)=0, \\
 (p-1)\left(p^{n-2} + \epsilon \eta_0(-1)\sqrt{p^*}^{n+s-2}\right),& \mbox{ if } g(\beta)\neq 0,
 \end{array}\right.
\ee
and  the number of such codewords $c_\beta$  follows from Lemma \ref{Lemma8}. Hence the proof is ended. 
\end{proof}
\begin{remark}
In Theorem \ref{Theorem1}, if $\epsilon \eta^{(n+s)/2}_0(-1)=-1$, then we need the  condition $0\leq s\leq n-4$; otherwise, $0\leq s\leq n-2$.
\end{remark}

 \begin{example}\label{Example0} The function $f:\F_{3^8}\to \F_{3}$ defined as $f(x)=\Tr^8( \zeta x^{4}  + \zeta^{816} x^2)$,  where $\F_{3^8}^{\star}=\langle \zeta \rangle$ with $\zeta^8 + 2\zeta^5 + \zeta^4+2\zeta^2+2\zeta+2=0$,   is the quadratic   $2$-plateaued unbalanced  function in the set  WRP with
 $$\Wa {f}(\beta)\in\{0,\epsilon \eta^5_0(-1)3^5\xi_3^{g(\beta)}\}=\{0,243,243\xi_3,243\xi_3^2\}$$
for all $\beta\in \F_{3^8}$, where $\epsilon=-1$,   $\eta_0(-1)=-1$ and $g$ is an unbalanced $3$-ary function with $g(0)=0$. 
 Then,   $\sC_{D_{f}}$  is the three-weight  linear code with parameters $[2348,8,1458
]_3$,  weight enumerator   $1+ 260y^{1458} + 5832y^{1566} +468y^{1620}$ and  weight distribution $(1,260,5832,468)$,  which is verified by MAGMA in \cite{bosma1997magma}. This code is minimal by Lemma \ref{Minimality}.
\end{example}

\begin{theorem}\label{Theorem2}
Let $f\in WRP$ and $n+s$ be an odd integer with $0\leq s\leq n-3$. Then, $\sC_{D_f}$ is the three-weight linear code with parameters 
$\left[p^{n-1}-1,n,(p-1)\left(p^{n-2}-   p^{(n+s-3)/2}\right)\right]_p$. The Hamming weights of the codewords and the weight distribution of $\sC_{D_f}$  are as in  Table \ref{table2}, where  $\epsilon=\pm 1$ is the sign of $\Wa{f}$.
\begin{table}[!htp]
\begin{center}	
\begin{tabular}{|c|c|c|}
\hline
Hamming weight $w$   & Multiplicity $A_w$ \\
\hline
\hline
\footnotesize{$0$} & \footnotesize{$1$}\\
\hline
\footnotesize{$ (p-1)p^{n-2}$} & \footnotesize{$p^n+p^{n-s-1}-p^{n-s} -1$}\\
\hline
\footnotesize{$(p-1)\left(p^{n-2}-\epsilon  \sqrt{p^*}^{n+s-3}\right)$} & \footnotesize{$\frac{p-1}{2}\left(p^{n-s-1}+\epsilon  \eta_0^{n}(-1)   \sqrt{p^*}^{n-s-1}\right)$}\\
\hline
\footnotesize{$(p-1)\left(p^{n-2}+\epsilon  \sqrt{p^*}^{n+s-3}\right)$} & \footnotesize{$\frac{p-1}{2}\left(p^{n-s-1} -\epsilon  \eta_0^{n}(-1)   \sqrt{p^*}^{n-s-1}\right)$}\\
\hline
\end{tabular}\end{center}	
\caption{
\label{table2}  The weight distribution of  $\sC_{D_f}$ when $n+s$ is odd}
\end{table}
\end{theorem}
\begin{proof}
The proof can be completed in a similar way to the even case in  Theorem \ref{Theorem1}. 
\end{proof}

 \begin{example}\label{Example0Odd} The function $f:\F_{3^3}\to \F_{3}$ defined as $f(x)=\Tr^3(    x^{4}  +   x^2)$,  where $\F_{3^3}^{\star}=\langle \zeta \rangle$ with $ \zeta^3+ 2\zeta +1=0$,   is the   quadratic bent function in the set  WRP with
 $$\Wa {f}(\beta)\in
\{i3\sqrt{3},i3\sqrt{3}\xi_3,i3\sqrt{3}\xi_3^2\}
=\{6\xi_3+3,-3\xi_3-6,-3\xi_3+3\}$$
for all $\beta\in \F_{3^3}$, where $\epsilon=-1$ and  $\eta_0(-1)=-1$. 
 Then,   $\sC_{D_{f}}$  is the three-weight   linear code with parameters $[8,3,4
]_3$,  weight enumerator   $1+8y^{6} + 6y^{8} +12y^{4}$ and  weight distribution $(1,8,6,12)$,  which is verified by MAGMA in \cite{bosma1997magma}. 
\end{example}

 Since  the Hamming weights of all nonzero codewords of  $\sC_{D_f}$ have a common divisor $p-1$, 
we can obtain a shorter linear code from the code  $\sC_{D_f}$.
Let $f\in WRP$.
 For any $x\in\F_q$,
$
f(x)=0$ if and only if $f(ax)=0,$ for every $a\in\F_p^{\star}$.
Then one can choose a subset $\bar{D}_f$ of the defining set $D_f$ of  $\sC_{D_f}$ defined by  (\ref{DefiningSet}) such that $\bigcup_{a\in \F_p^{\star}}a\bar{D}_f$ is a partition of $D_f$, namely,
\be\nn
D_f=\F_p^{\star}\bar{D}_f=\{ab :a\in\F_p^{\star}, b\in  \bar{D}_f\},
\ee
where for each pair of distinct elements $b_1,b_2\in \bar{D}_f$ we have $\frac{b_1}{b_2}\notin\F_p^{\star}$.
 This implies that the linear code $\sC_{D_f}$ can be punctured into a shorter linear code $\sC_{\bar{D}_f}$, where $\bar{D}_f$ is its defining set. Notice that for $\beta\in\F_q^{\star}$,  
\be\nn
 \#\{x\in D_f : f(x)=0 \mbox{ and } \Tr^n(\beta x)=0\}=(p-1) \#\{x\in \bar{D}_f : f(x)=0 \mbox{ and } \Tr^n(\beta x)=0\}.
\ee
Hence, the following linear codes  in Corollaries \ref{Corollary1}  and  \ref{Corollary2} are directly obtained   from the constructed ones  in Theorems  \ref{Theorem1} and  \ref{Theorem2}, respectively.

\begin{corollary}\label{Corollary1}
The punctured version $\sC_{\bar{D}_f}$  of the   code $\sC_{D_f}$ of Theorem \ref{Theorem1} is the three-weight linear code with parameters 
$ \left[ (p^{n-1}-1)/(p-1)+\epsilon \eta_0(-1) \sqrt{p^*}^{n+s-2},n\right]_p$ whose weight distribution is listed in  Table \ref{table111}. 
\begin{table}[!htp]
\begin{center}	
\begin{tabular}{|c|c|c|}
\hline
Hamming weight $w$   & Multiplicity $A_w$ \\
\hline
\hline
\footnotesize{$0$} & \footnotesize{$1$}\\
\hline
\footnotesize{$ p^{n-2}+\epsilon (p-1)\sqrt{p^*}^{n+s-4} $} & \footnotesize{$p^n-p^{n-s}$}\\
\hline
\footnotesize{$  p^{n-2}$} & \footnotesize{$p^{n-s-1} + \epsilon \eta_0^{n+1}(-1) (p-1) \sqrt{p^*}^{n-s-2} -1$}\\
\hline
\footnotesize{$ p^{n-2} + \epsilon \eta_0(-1)\sqrt{p^*}^{n+s-2} $} & \footnotesize{$(p-1)\left(p^{n-s-1} -\epsilon \eta_0^{n+1}(-1) \sqrt{p^*}^{n-s-2}\right)$}\\
\hline
\end{tabular}\end{center}	
\caption{
\label{table111}  The weight distribution of  $\sC_{\bar{D}_f}$ when $n+s$ is even}
\end{table}
\end{corollary}
 \begin{example} The punctured version $\sC_{\bar{D}_{f}}$ of $\sC_{D_{f}}$ in  Example  \ref{Example0}  is the three-weight   linear  code with parameters $[1174,8,729]_3$,  weight enumerator   $1+ 260y^{729} + 5832y^{783} +468y^{810}$ and  weight distribution $(1,260,5832,468)$. This code is minimal by Lemma \ref{Minimality}.  
\end{example}

\begin{corollary}\label{Corollary2}
The punctured version $\sC_{\bar{D}_f}$  of the  code $\sC_{D_f}$ of  Theorem \ref{Theorem2}  is the three-weight linear code with parameters 
$\left[ (p^{n-1}-1)/(p-1),n,p^{n-2}-   p^{(n+s-3)/2} \right]_p$ whose weight distribution is listed in  Table \ref{table222}. 
\begin{table}[!htp]
\begin{center}	
\begin{tabular}{|c|c|c|}
\hline
Hamming weight $w$   & Multiplicity $A_w$ \\
\hline
\hline
\footnotesize{$0$} & \footnotesize{$1$}\\
\hline
\footnotesize{$ p^{n-2}$} & \footnotesize{$p^n+p^{n-s-1}-p^{n-s} -1$}\\
\hline
\footnotesize{$ p^{n-2}-\epsilon  \sqrt{p^*}^{n+s-3} $} & \footnotesize{$\frac{p-1}{2}\left(p^{n-s-1}+\epsilon  \eta_0^{n}(-1)   \sqrt{p^*}^{n-s-1}\right)$}\\
\hline
\footnotesize{$ p^{n-2}+\epsilon  \sqrt{p^*}^{n+s-3} $} & \footnotesize{$\frac{p-1}{2}\left(p^{n-s-1} -\epsilon  \eta_0^{n}(-1)   \sqrt{p^*}^{n-s-1}\right)$}\\
\hline
\end{tabular}\end{center}	
\caption{
\label{table222}  The weight distribution of  $\sC_{\bar{D}_f}$ when $n+s$ is odd}
\end{table}
\end{corollary}

 \begin{example} The punctured version $\sC_{\bar{D}_{f}}$ of $\sC_{D_{f}}$   in Example  \ref{Example0Odd}  is the three-weight   linear   code with parameters $[4,3,2]_3$,  weight enumerator   $1+8y^{3} + 6y^{4} +12y^2$ and  weight distribution $(1,8,6,12)$,  which is verified by MAGMA in \cite{bosma1997magma}.  This code is   optimal owing to the Singleton bound. 
\end{example}

In particular, we can work on the Walsh support of a weakly regular plateaued function $f$ to define a subcode  of  each constructed code above.
We   consider a linear code involving $D_f$ defined by
\be\label{LinearSubCodeS}
\bar{\sC}_{D_f}=\{c_\beta= (\Tr^n(\beta d_1), \Tr^n(\beta d_2), \ldots, \Tr^n(\beta d_m)): \beta \in Supp(\Wa f )\},
\ee
which is the subcode  of $\sC_{D_f}$ defined by (\ref{LinearCodes}). Hence, the following codes   in  Corollaries \ref{Corollary3}, \ref{Corollary4}, \ref{Corollary5} and \ref{Corollary6} are the subcodes of the codes of 
Theorems \ref{Theorem1}, \ref{Theorem2} and Corollaries \ref{Corollary1}, \ref{Corollary2}, respectively. Notice that their parameters are directly derived from that of corresponding codes.

\begin{corollary}\label{Corollary3}
The subcode  $\bar{\sC}_{D_f}$ of the code $\sC_{D_f}$ of Theorem \ref{Theorem1} is the two-weight linear code with parameters 
$\left[ p^{n-1}-1+\epsilon \eta_0(-1)(p-1) \sqrt{p^*}^{n+s-2},n-s\right]_p$ whose weight distribution is listed in  Table \ref{table11}. 
\begin{table}[!htp]
\begin{center}	
\begin{tabular}{|c|c|c|}
\hline
Hamming weight $w$   & Multiplicity $A_w$ \\
\hline
\hline
\footnotesize{$0$} & \footnotesize{$1$}\\
\hline
\footnotesize{$ (p-1)p^{n-2}$} & \footnotesize{$p^{n-s-1} + \epsilon \eta_0^{n+1}(-1) (p-1) \sqrt{p^*}^{n-s-2} -1$}\\
\hline
\footnotesize{$(p-1)\left(p^{n-2} + \epsilon \eta_0(-1)\sqrt{p^*}^{n+s-2}\right)$} & \footnotesize{$(p-1)\left(p^{n-s-1} -\epsilon \eta_0^{n+1}(-1) \sqrt{p^*}^{n-s-2}\right)$}\\
\hline
\end{tabular}\end{center}	
\caption{
\label{table11}  The weight distribution of  $\bar{\sC}_{D_f}$ when $n+s$ is even}
\end{table}
\end{corollary}

\begin{corollary}\label{Corollary4}
The subcode  $\bar{\sC}_{D_f}$ of the code $\sC_{D_f}$ of Theorem \ref{Theorem2} is the three-weight linear code with parameters 
$\left[p^{n-1}-1,n-s,(p-1)\left(p^{n-2}-   p^{(n+s-3)/2}\right)\right]_p$ whose weight distribution is listed in  Table \ref{table22}. 
\begin{table}[!htp]
\begin{center}	
\begin{tabular}{|c|c|c|}
\hline
Hamming weight $w$   & Multiplicity $A_w$ \\
\hline
\hline
\footnotesize{$0$} & \footnotesize{$1$}\\
\hline
\footnotesize{$ (p-1)p^{n-2}$} & \footnotesize{$p^{n-s-1} -1$}\\
\hline
\footnotesize{$(p-1)\left(p^{n-2}-\epsilon  \sqrt{p^*}^{n+s-3}\right)$} & \footnotesize{$\frac{p-1}{2}\left(p^{n-s-1}+\epsilon  \eta_0^{n}(-1)   \sqrt{p^*}^{n-s-1}\right)$}\\
\hline
\footnotesize{$(p-1)\left(p^{n-2}+\epsilon  \sqrt{p^*}^{n+s-3}\right)$} & \footnotesize{$\frac{p-1}{2}\left(p^{n-s-1} -\epsilon  \eta_0^{n}(-1)   \sqrt{p^*}^{n-s-1}\right)$}\\
\hline
\end{tabular}\end{center}	
\caption{
\label{table22}  The weight distribution of  $\bar{\sC}_{D_f}$ when $n+s$ is odd}
\end{table}
\end{corollary}

\begin{corollary}\label{Corollary5}
The subcode  $\bar{\sC}_{\bar{D}_f}$ of the code $\sC_{\bar{D}_f}$ of  Corollary \ref{Corollary1}    is the two-weight linear code with parameters 
$ \left[ (p^{n-1}-1)/(p-1)+\epsilon \eta_0(-1) \sqrt{p^*}^{n+s-2},n-s\right]_p$ whose weight distribution is listed in  Table \ref{table1111}. 
\begin{table}[!htp]
\begin{center}	
\begin{tabular}{|c|c|c|}
\hline
Hamming weight $w$   & Multiplicity $A_w$ \\
\hline
\hline
\footnotesize{$0$} & \footnotesize{$1$}\\
\hline
\footnotesize{$  p^{n-2}$} & \footnotesize{$p^{n-s-1} + \epsilon \eta_0^{n+1}(-1) (p-1) \sqrt{p^*}^{n-s-2} -1$}\\
\hline
\footnotesize{$ p^{n-2} + \epsilon \eta_0(-1)\sqrt{p^*}^{n+s-2} $} & \footnotesize{$(p-1)\left(p^{n-s-1} -\epsilon \eta_0^{n+1}(-1) \sqrt{p^*}^{n-s-2}\right)$}\\
\hline
\end{tabular}\end{center}	
\caption{
\label{table1111}  The weight distribution of  $\bar{\sC}_{\bar{D}_f}$ when $n+s$ is even}
\end{table}
\end{corollary}
\begin{corollary}\label{Corollary6}
The subcode  $\bar{\sC}_{\bar{D}_f}$ of the code $\sC_{\bar{D}_f}$ of  Corollary \ref{Corollary2}  is the three-weight linear code with parameters 
$\left[ (p^{n-1}-1)/(p-1),n-s, p^{n-2}-  p^{(n+s-3)/2} \right]_p$ whose weight distribution is listed in  Table \ref{table2222}. 
\begin{table}[!htp]
\begin{center}	
\begin{tabular}{|c|c|c|}
\hline
Hamming weight $w$   & Multiplicity $A_w$ \\
\hline
\hline
\footnotesize{$0$} & \footnotesize{$1$}\\
\hline
\footnotesize{$ p^{n-2}$} & \footnotesize{$ p^{n-s-1}  -1$}\\
\hline
\footnotesize{$ p^{n-2}-\epsilon  \sqrt{p^*}^{n+s-3} $} & \footnotesize{$\frac{p-1}{2}\left(p^{n-s-1}+\epsilon  \eta_0^{n}(-1)   \sqrt{p^*}^{n-s-1}\right)$}\\
\hline
\footnotesize{$ p^{n-2}+\epsilon  \sqrt{p^*}^{n+s-3} $} & \footnotesize{$\frac{p-1}{2}\left(p^{n-s-1} -\epsilon  \eta_0^{n}(-1)   \sqrt{p^*}^{n-s-1}\right)$}\\
\hline
\end{tabular}\end{center}	
\caption{
\label{table2222}  The weight distribution of  $\bar{\sC}_{\bar{D}_f}$ when $n+s$ is odd}
\end{table}
\end{corollary}

  \begin{remark}
When we  assume only the quadratic bent-ness    (resp., the weakly regular bent-ness) in this subsection, we can obviously recover the linear codes obtained by Zhou et al.   \cite{zhou2016linear} (resp.,  
 by Tang et al.   \cite{tang2016linear}).  Therefore, this subsection can be viewed as an extension of \cite{zhou2016linear} and   \cite{tang2016linear} to the weakly regular  plateaued unbalanced functions.
\end{remark}
The following section, to construct new linear codes,  pushes further the use of weakly regular plateaued functions in the construction methods proposed by Tang et al.  \cite{tang2016linear}.  
\subsection{Two or  three weight linear codes with their weight distributions}

Let $f:\F_q\to\F_p$ be a $p$-ary function. Define the sets
\be\nn
\begin{array}{ll}
D_{f,sq}=\{ x\in\F_q: f(x)\in SQ\} \mbox{ and } D_{f,nsq}=\{ x\in\F_q: f(x)\in NSQ\}.
\end{array}
\ee
With the similar definition of the linear code $\sC_{D_f}$ defined by (\ref{LinearCodes}), we can define a linear code    involving $D_{f,sq}=\{d'_1, d'_2, \ldots, d'_m\}$ 
\be\label{LinearCodeSQ}
\sC_{D_{f,sq}}=\{c_\beta= (\Tr{^n}(\beta d'_1), \Tr{^n}(\beta d'_2), \ldots, \Tr{^n}(\beta d'_m)) : \beta \in\mathbb{F}_q\}
\ee
and 
a  linear code involving $D_{f,nsq}=\{d''_1, d''_2, \ldots, d''_m\}$
\be\label{LinearCodeNSQ}
\sC_{D_{f,nsq}}=\{c_\beta= (\Tr{^n}(\beta d''_1), \Tr{^n}(\beta d''_2), \ldots, \Tr{^n}(\beta d''_m)) : \beta \in\mathbb{F}_q\}.
\ee
From  Lemmas \ref{Lemma7} and \ref{Lemma14}, we find  the Hamming weights of the codewords of the linear codes $\sC_{D_{f,sq}}$ and $\sC_{D_{f,nsq}}$ as well as their length, and we  determine their weight distributions from Lemmas \ref{SupportLemma} and \ref{Lemma8}.

\begin{theorem}\label{Theorem3}
Let $n+s$ be an even integer and   $f\in WRP$. Then, $\sC_{D_{f,sq}}$ is  the three-weight linear code with parameters 
$\left[\frac{p-1}{2}\left(p^{n-1}-\epsilon \eta_0(-1) \sqrt{p^*}^{n+s-2}\right), n\right]_p$, where  $\epsilon=\pm 1$ is the sign of $\Wa{f}$. The Hamming weights of the codewords and the weight distribution of $\sC_{D_{f,sq}}$  are as in  Table \ref{table3}.
\begin{table}[!htp]
\begin{center}	
\begin{tabular}{|c|c|c|}
\hline
Hamming weight $w$   & Multiplicity $A_w$ \\
\hline
\hline
\footnotesize{$0$} & \footnotesize{$1$}\\
\hline
\footnotesize{$\frac{(p-1)^2}{2}\left(p^{n-2}-\epsilon  \sqrt{p^*}^{n+s-4}\right)$} & \footnotesize{$p^n-p^{n-s}$}\\
\hline
\footnotesize{$  \frac{(p-1)^2}{2}p^{n-2}$} 
& \footnotesize{$ p^{n-s-1} +\frac{p-1}{2}\left(p^{n-s-1}+\epsilon \eta_0^{n+1}(-1)\sqrt{p^*}^{n-s-2}\right)-1 $}\\
\hline
\footnotesize{$(p-1)\left(\frac{p-1}{2}p^{n-2}- \epsilon \eta_0(-1) \sqrt{p^*}^{n+s-2}\right)$}
 & \footnotesize{$\frac{p-1}{2}\left(p^{n-s-1} -\epsilon \eta_0^{n+1}(-1)\sqrt{p^*}^{n-s-2}\right)$}\\
\hline
\end{tabular}\end{center}	
\caption{
\label{table3}  The weight distribution of   $\sC_{D_{f,sq}}$ when $n+s$ is even}
\end{table}
\end{theorem}

\begin{proof}
We have  $\# D_{f,sq}=\frac{p-1}{2}(p^{n-1}-\epsilon \eta_0(-1) \sqrt{p^*}^{n+s-2})$ by Lemma \ref{Lemma7}  and   $\wt(c_\beta)=\# D_{f,sq}- \sN_{sq,\beta}$ for every $\beta\in\F_q^{\star}$ by Lemma \ref{Lemma14}. Then,  for every $\beta\in\F_q^{\star}\setminus \Supp(\Wa f),$  
\be\nn\begin{array}{ll}
\wt(c_\beta)&=\frac{(p-1)^2}{2}\left(p^{n-2}-\epsilon  \sqrt{p^*}^{n+s-4}\right),
\end{array}
\ee
and  the number of such codewords $c_\beta$ is equal to $p^n-p^{n-s}$ by Lemma \ref{SupportLemma}.
 For every  $0\neq  \beta\in \Supp(\Wa f),$ 
\be\nn\begin{array}{ll}
\wt(c_\beta)&=\left\{\begin{array}{ll}
\frac{(p-1)^2}{2}p^{n-2},  & \mbox{ if } g(\beta)=0 \mbox{ or } g(\beta) \in NSQ, \\
\frac{(p-1)^2}{2}p^{n-2} - \epsilon \eta_0(-1)(p-1) \sqrt{p^*}^{n+s-2},& \mbox{ if } g(\beta) \in SQ,
 \end{array}\right.
\end{array}
\ee
and  the number of such codewords $c_\beta$  follows from Lemma \ref{Lemma8}. Hence the proof is ended. 
\end{proof}
\begin{remark}
In Theorem \ref{Theorem3}, if $\epsilon \eta^{(n+s)/2}_0(-1)=1$ and $p=3$, then we  have the  condition $0\leq s\leq n-4$; otherwise, $0\leq s\leq n-2$ and $n\geq 3$.
\end{remark}

 \begin{example}\label{ExampleSQ} The function $f:\F_{3^5}\to \F_{3}$ defined as $f(x)=\Tr^5( \zeta^{19}x^{4}  + \zeta^{238} x^2)$,  where $\F_{3^5}^{\star}=\langle \zeta \rangle$ with $\zeta^5 + 2\zeta  +1=0$,   is the quadratic $1$-plateaued unbalanced  function in the set  WRP with
 $$\Wa {f}(\beta)\in\{0,\epsilon \eta^3_0(-1)3^3\xi_3^{g(\beta)}\}=\{0,-27,-27\xi_3,-27\xi_3^2\}$$
for all $\beta\in \F_{3^5}$, where $\epsilon=1$,   $\eta_0(-1)=-1$ and $g$ is an unbalanced $3$-ary function with $g(0)=0$. 
 Then,   $\sC_{D_{f,sq}}$  is the three-weight   linear code with parameters $[90,5,54
]_3$,  weight enumerator   $1+ 50y^{54} + 162y^{60} +30y^{72}$ and  weight distribution $(1,50,162,30)$,  which is verified by MAGMA in \cite{bosma1997magma}. This code is minimal by Lemma \ref{Minimality}. 
\end{example}

Recall that we have the following fact:
\be\nn
\eta_0(-1)= \left\{ \begin{array}{ll}
\,\,\,\, 1 & \mbox{ if and only if } \, \, p \equiv 1\pmod 4,\\
-1 &  \mbox{ if and only if }  \,  p \equiv 3\pmod 4,
\end{array} \right.
\ee
which will be needed   
in Theorems \ref{Theorem4} and \ref{Theorem5}.

\begin{theorem}\label{Theorem4}
Let  $n+s$ be an odd  integer  and $f\in WRP$. Then, $\sC_{D_{f,sq}}$ is the three-weight linear code with parameters 
$\left[\frac{p-1}{2}\left(p^{n-1}+\epsilon \sqrt{p^*}^{n+s-1}\right),n\right]_p$, where  $\epsilon=\pm 1$ is the sign of $\Wa{f}$. The Hamming weights of the codewords and the weight distribution of $\sC_{D_{f,sq}}$ are as in  Table \ref{table5} and Table \ref{table5.} when $p \equiv 1\pmod 4$  and $p \equiv 3\pmod 4$, respectively.
\begin{table}[!htp]
\begin{center}	
\begin{tabular}{|c|c|c|}
\hline
Hamming weight $w$   & Multiplicity $A_w$ \\
\hline
\hline
\footnotesize{$0$} & \footnotesize{$1$}\\
\hline
\footnotesize{$  \frac{(p-1)^2}{2}p^{n-2}$} & \footnotesize{$p^{n-s-1}-1   $}\\
\hline
\footnotesize{$ \frac{p-1}{2}\left((p-1)p^{n-2}+\epsilon (p+1) \sqrt{p}^{n+s-3}\right) $} & \footnotesize{$ \frac{p-1}{2}\left( p^{n-s-1} + \epsilon     \sqrt{p}^{n-s-1}\right) $}\\
\hline
\footnotesize{ $\frac{(p-1)^2}{2}\left(p^{n-2}+\epsilon \sqrt{p}^{n+s-3}\right) $}& \footnotesize{$p^n-p^{n-s}+\frac{p-1}{2}\left( p^{n-s-1} - \epsilon   \sqrt{p}^{n-s-1}\right)  $}\\
\hline
\end{tabular}\end{center}	
\caption{
\label{table5}  The weight distribution of  $\sC_{D_{f,sq}}$  when $p \equiv 1\pmod 4$ and $n+s$ is odd}
\begin{center}	
\begin{tabular}{|c|c|c|}
\hline
Hamming weight $w$   & Multiplicity $A_w$ \\
\hline
\hline
\footnotesize{$0$} & \footnotesize{$1$}\\
\hline
\footnotesize{$  \frac{(p-1)^2}{2}p^{n-2}$} & \footnotesize{$p^{n-s-1}-1   $}\\
\hline
\footnotesize{$ \frac{(p-1)^2}{2}\left( p^{n-2}-\epsilon   \sqrt{p^*}^{n+s-3}\right) $} & \footnotesize{$p^n-p^{n-s}+ \frac{p-1}{2}\left( p^{n-s-1} + \epsilon (-1)^n    \sqrt{p^*}^{n-s-1}\right) $}\\
\hline
\footnotesize{ $\frac{p-1}{2}\left((p-1)p^{n-2}-\epsilon (p+1)\sqrt{p^*}^{n+s-3}\right) $}& \footnotesize{$\frac{p-1}{2}\left( p^{n-s-1} - \epsilon  (-1)^n   \sqrt{p^*}^{n-s-1}\right)  $}\\
\hline
\end{tabular}\end{center}	
\caption{
\label{table5.}  The weight distribution of  $\sC_{D_{f,sq}}$  when $p \equiv 3\pmod 4$ and $n+s$ is odd}
\end{table}

\end{theorem}

\begin{proof} The proof can be completed in a similar way to the even case in  Theorem \ref{Theorem3}. 
\end{proof}
\begin{remark}
In Theorem \ref{Theorem4},  if  $p \equiv 3\pmod 4$ and $\epsilon \eta^{(n+s-1)/2}_0(-1)=-1$ or   $p \equiv 1\pmod 4$ and $\epsilon=-1$, then we  have the  condition $0\leq s\leq n-3$; otherwise, $0\leq s\leq n-1$ and $n\geq 2$.
\end{remark}

\begin{theorem}\label{Theorem5}
Let   $n+s$ be an odd  integer and $f\in WRP$. Then, $\sC_{D_{f,nsq}}$ is the   three-weight linear code with parameters 
$\left[\frac{p-1}{2}\left(p^{n-1}-\epsilon \sqrt{p^*}^{n+s-1}\right),n\right]_p$, where  $\epsilon=\pm 1$ is the sign of $\Wa{f}$. The Hamming weights of the codewords and the weight distribution of $\sC_{D_{f,nsq}}$ are as in  Table \ref{table4} and Table \ref{table4.} when $p \equiv 1\pmod 4$  and $p \equiv 3\pmod 4$, respectively.
\begin{table}[!htp]
\begin{center}	
\begin{tabular}{|c|c|c|}
\hline
Hamming weight $w$   & Multiplicity $A_w$ \\
\hline
\hline
\footnotesize{$0$} & \footnotesize{$1$}\\
\hline
\footnotesize{$\frac{(p-1)^2}{2}p^{n-2}  $} & \footnotesize{$p^{n-s-1}-1   $}\\
\hline
\footnotesize{$ \frac{(p-1)^2}{2}\left( p^{n-2}-\epsilon \sqrt{p}^{n+s-3}\right) $} & \footnotesize{$p^n-p^{n-s}+\frac{p-1}{2}\left(p^{n-s-1} +\epsilon    \sqrt{p}^{n-s-1} \right)$}\\
\hline
\footnotesize{ $ \frac{p-1}{2}\left((p-1)p^{n-2}-\epsilon(p+1) \sqrt{p}^{n+s-3}\right)$}& \footnotesize{$\frac{p-1}{2}\left( p^{n-s-1} - \epsilon     \sqrt{p}^{n-s-1}\right)  $}\\
\hline
\end{tabular}\end{center}	
\caption{
\label{table4}  The weight distribution of  $\sC_{D_{f,nsq}}$  when $p \equiv 1\pmod 4$ and $n+s$ is odd}

\begin{center}	
\begin{tabular}{|c|c|c|}
\hline
Hamming weight $w$   & Multiplicity $A_w$ \\
\hline
\hline
\footnotesize{$0$} & \footnotesize{$1$}\\
\hline
\footnotesize{$\frac{(p-1)^2}{2}p^{n-2}  $} & \footnotesize{$p^{n-s-1}-1   $}\\
\hline
\footnotesize{$ \frac{p-1}{2}\left((p-1)p^{n-2}+\epsilon(p+1) \sqrt{p^*}^{n+s-3}\right) $} & \footnotesize{$\frac{p-1}{2}\left(p^{n-s-1} +\epsilon   (-1)^n    \sqrt{p^*}^{n-s-1} \right)$}\\
\hline
\footnotesize{ $ \frac{(p-1)^2}{2}\left(p^{n-2}+\epsilon  \sqrt{p^*}^{n+s-3}\right)$}& \footnotesize{$p^n-p^{n-s}+\frac{p-1}{2}\left( p^{n-s-1} - \epsilon  (-1)^n \sqrt{p^*}^{n-s-1}\right)  $}\\
\hline
\end{tabular}\end{center}	
\caption{
\label{table4.}  The weight distribution of  $\sC_{D_{f,nsq}}$  when $p \equiv 3\pmod 4$ and $n+s$ is odd}
\end{table}
\end{theorem}
\begin{proof}
Obviously, 
 we get 
 $\# D_{f,nsq}=\frac{p-1}{2}(p^{n-1}-\epsilon \sqrt{p^*}^{n+s-1})$ by Lemma \ref{Lemma7}
and $\wt(c_\beta)=\#  D_{f,nsq}- \sN_{nsq,\beta}$ for every $\beta \in\F_q^{\star}$ by Lemma \ref{Lemma14}. Then,  for every $\beta\in\F_q^{\star}\setminus \Supp(\Wa f),$   
\be\nn\begin{array}{ll}
\wt(c_\beta)&=\frac{(p-1)^2}{2}\left(p^{n-2}-\epsilon \eta_0(-1)\sqrt{p^*}^{n+s-3}\right)
\end{array}
\ee
and  the number of such codewords $c_\beta$ is equal to $p^n-p^{n-s}$ by Lemma \ref{SupportLemma}.
 For every  $0\neq  \beta\in \Supp(\Wa f),$
\be\nn\begin{array}{ll}
\wt(c_\beta)&=\left\{\begin{array}{ll}
\frac{p-1}{2}(p-1)p^{n-2},  & \mbox{ if } g(\beta)=0, \\
\frac{p-1}{2}\left((p-1)p^{n-2}+\epsilon \sqrt{p^*}^{n+s-3}(1-p^*)\right),& \mbox{ if } g(\beta) \in SQ,\\
\frac{p-1}{2}\left((p-1)p^{n-2}-\epsilon \sqrt{p^*}^{n+s-3}(p^*+1)\right),  & \mbox{ if }  g(\beta) \in NSQ,
 \end{array}\right.
 \end{array}
\ee and  the number of such codewords $c_\beta$  follows from Lemma \ref{Lemma8}.  Hence the proof is ended.  
\end{proof}
\begin{remark}
In Theorem \ref{Theorem5}, if   $p \equiv 3\pmod 4$ and $\epsilon \eta^{(n+s-1)/2}_0(-1)=1$ or  $p \equiv 1\pmod 4$ and $\epsilon=1$, then we have   the  condition $0\leq s\leq n-3$; otherwise, $0\leq s\leq n-1$ and $n\geq 2$.
\end{remark}
 \begin{example}\label{ExampleNSQ} The function $f:\F_{3^6}\to \F_{3}$ defined as $f(x)=\Tr^6(  \zeta x^{4}  + \zeta^{27} x^2)$,  where $\F_{3^6}^{\star}=\langle \zeta \rangle$ with $\zeta^6 + 2\zeta^4  +\zeta^2+2\zeta +2=0$,   is the  quadratic $1$-plateaued unbalanced  function in the set  WRP with
 $$\Wa {f}(\beta)\in 
\{0,i27\sqrt{3},i27\sqrt{3}\xi_3,i27\sqrt{3}\xi_3^2\}
=\{0,54\xi_3+27,-27\xi_3-54,-27\xi_3+27\}$$
for all $\beta\in \F_{3^6}$, where $\epsilon=-1$,   $\eta_0(-1)=-1$ and $g$ is an unbalanced $3$-ary function with $g(0)=0$. 
 Then,   $\sC_{D_{f,nsq}}$  is the three-weight  linear code with parameters $[216,6,126
]_3$,  weight enumerator   $1+ 72y^{126} + 576y^{144} +80y^{162}$ and  weight distribution $(1,72,576,80)$,  which is verified by MAGMA in \cite{bosma1997magma}. This code is minimal by Lemma \ref{Minimality}.  
\end{example}
\begin{remark} Let $n+s$ be an even and $f\in WRP$. Then, $\sC_{D_{f,nsq}}$ is  the three-weight linear code with the same parameters and weight distribution of $\sC_{D_{f,sq}}$ in Theorem \ref{Theorem3}.
\end{remark}

We now  obtain a shorter linear code from the above each constructed code. 
Let $f\in WRP$. For any $x\in\F_q$,
$f(x)$  is a quadratic residue (resp., quadratic non-residue) in $\F_p^{\star}$ if and only if  $f(ax)$  is a quadratic residue (resp., quadratic non-residue) in $\F_p^{\star}$ for every $a\in\F_p^{\star}$.
Then one can choose a subset $\bar{D}_{f,sq}$ of the defining set  $D_{f,sq}$ of  $\sC_{D_{f,sq}}$ such that 
\be\nn
D_{f,sq}=\F_p^{\star}\bar{D}_{f,sq}=\{ab :a\in\F_p^{\star}, b\in  \bar{D}_{f,sq}\},
\ee
and  a subset $\bar{D}_{f,nsq}$ of the defining set  $D_{f,nsq}$ of  $\sC_{D_{f,nsq}}$ such that
$
D_{f,nsq}=\{ab :a\in\F_p^{\star}, b\in  \bar{D}_{f,nsq}\}.
$
 Hence, one can easily obtain the punctured versions $\sC_{\bar{D}_{f,sq}}$  and $\sC_{\bar{D}_{f,nsq}}$ of $\sC_{D_{f,sq}}$ and  $\sC_{D_{f,nsq}}$, respectively, whose parameters are  derived directly from that of the original codes.
Notice that  Corollaries \ref{Corollary7}, \ref{Corollary8}  and  \ref{Corollary9} follow directly from Theorems  \ref{Theorem3}, \ref{Theorem4} and  \ref{Theorem5}, respectively.

\begin{corollary}\label{Corollary7} 
The punctured version  $\sC_{\bar{D}_{f,sq}}$  of the    code  $\sC_{D_{f,sq}}$ of Theorem \ref{Theorem3}   is  the three-weight linear code with parameters 
$\left[\frac{1}{2}(p^{n-1}-\epsilon \eta_0(-1) \sqrt{p^*}^{n+s-2}), n\right]_p$ whose weight distribution is listed in  Table \ref{table333}. 
\begin{table} 
\begin{center}	
\begin{tabular}{|c|c|c|}
\hline
Hamming weight $w$   & Multiplicity $A_w$ \\
\hline
\hline
\footnotesize{$0$} & \footnotesize{$1$}\\
\hline
\footnotesize{$\frac{p-1}{2}\left(p^{n-2}-\epsilon  \sqrt{p^*}^{n+s-4}\right)$} & \footnotesize{$p^n-p^{n-s}$}\\
\hline
\footnotesize{$  \frac{(p-1)}{2}p^{n-2}$} 
& \footnotesize{$ p^{n-s-1} +\frac{p-1}{2}\left(p^{n-s-1}+\epsilon \eta_0^{n+1}(-1)\sqrt{p^*}^{n-s-2}\right)-1 $}\\
\hline
\footnotesize{$ \frac{p-1}{2}p^{n-2}- \epsilon \eta_0(-1) \sqrt{p^*}^{n+s-2} $}
 & \footnotesize{$\frac{p-1}{2}\left(p^{n-s-1} -\epsilon \eta_0^{n+1}(-1)\sqrt{p^*}^{n-s-2}\right)$}\\
\hline
\end{tabular}\end{center}	
\caption{
\label{table333}  The weight distribution of   $\sC_{\bar{D}_{f,sq}}$ when $n+s$ is even}
\end{table}
\end{corollary}

 \begin{example} The punctured version $\sC_{\bar{D}_{f,sq}}$ of $\sC_{D_{f,sq}}$  in Example  \ref{ExampleSQ}  is the three-weight    linear  code with parameters $[45,5,27]_3$,  weight enumerator   $1+ 50y^{27} + 162y^{30} +30y^{36}$ and  weight distribution $(1,50,162,30)$. 
This code is minimal by Lemma \ref{Minimality} and is  almost optimal owing to the Griesmer bound. 
\end{example}

\begin{corollary}\label{Corollary8}
The punctured version  $\sC_{\bar{D}_{f,sq}}$  of the    code  $\sC_{D_{f,sq}}$  of Theorem \ref{Theorem4}  is the three-weight linear code with parameters 
$\left[\frac{1}{2}(p^{n-1}+\epsilon \sqrt{p^*}^{n+s-1}),n\right]_p$ whose weight distribution is listed in  Table \ref{table555} and Table \ref{table555.} when $p \equiv 1\pmod 4$  and $p \equiv 3\pmod 4$, respectively.
\begin{table}[!htp]
\begin{center}	
\begin{tabular}{|c|c|c|}
\hline
Hamming weight $w$   & Multiplicity $A_w$ \\
\hline
\hline
\footnotesize{$0$} & \footnotesize{$1$}\\
\hline
\footnotesize{$  \frac{p-1}{2}p^{n-2}$} & \footnotesize{$p^{n-s-1}-1   $}\\
\hline
\footnotesize{$ \frac{1}{2}\left((p-1)p^{n-2}+\epsilon (p+1) \sqrt{p}^{n+s-3}\right) $} & \footnotesize{$ \frac{p-1}{2}\left( p^{n-s-1} + \epsilon     \sqrt{p}^{n-s-1}\right) $}\\
\hline
\footnotesize{ $\frac{(p-1)}{2}\left(p^{n-2}+\epsilon \sqrt{p}^{n+s-3}\right) $}& \footnotesize{$p^n-p^{n-s}+\frac{p-1}{2}\left( p^{n-s-1} - \epsilon   \sqrt{p}^{n-s-1}\right)  $}\\
\hline
\end{tabular}\end{center}	
\caption{
\label{table555}  The weight distribution of  $\sC_{\bar{D}_{f,sq}}$  when $p \equiv 1\pmod 4$ and $n+s$ is odd}
\begin{center}	
\begin{tabular}{|c|c|c|}
\hline
Hamming weight $w$   & Multiplicity $A_w$ \\
\hline
\hline
\footnotesize{$0$} & \footnotesize{$1$}\\
\hline
\footnotesize{$  \frac{p-1}{2}p^{n-2}$} & \footnotesize{$p^{n-s-1}-1   $}\\
\hline
\footnotesize{$ \frac{p-1}{2}\left( p^{n-2}-\epsilon   \sqrt{p^*}^{n+s-3}\right) $} & \footnotesize{$p^n-p^{n-s}+ \frac{p-1}{2}\left( p^{n-s-1} + \epsilon  (-1)^n    \sqrt{p^*}^{n-s-1}\right) $}\\
\hline
\footnotesize{ $\frac{1}{2}\left((p-1)p^{n-2}-\epsilon (p+1)\sqrt{p^*}^{n+s-3}\right) $}& \footnotesize{$\frac{p-1}{2}\left( p^{n-s-1} - \epsilon  (-1)^n   \sqrt{p^*}^{n-s-1}\right)  $}\\
\hline
\end{tabular}\end{center}	
\caption{
\label{table555.}  The weight distribution of  $\sC_{\bar{D}_{f,sq}}$  when $p \equiv 3\pmod 4$ and $n+s$ is odd}
\end{table}
\end{corollary}

\begin{corollary}\label{Corollary9}
The punctured version   $\sC_{\bar{D}_{f,nsq}}$   of the    code  $\sC_{D_{f,nsq}}$ of Theorem \ref{Theorem5}  is the   three-weight linear code with parameters 
$\left[\frac{1}{2}(p^{n-1}-\epsilon \sqrt{p^*}^{n+s-1}),n\right]_p$ whose weight distribution is listed in  Table \ref{table444} and Table \ref{table444.} when $p \equiv 1\pmod 4$  and $p \equiv 3\pmod 4$, respectively.
\begin{table}[!htp]
\begin{center}	
\begin{tabular}{|c|c|c|}
\hline
Hamming weight $w$   & Multiplicity $A_w$ \\
\hline
\hline
\footnotesize{$0$} & \footnotesize{$1$}\\
\hline
\footnotesize{$\frac{(p-1)}{2}p^{n-2}  $} & \footnotesize{$p^{n-s-1}-1   $}\\
\hline
\footnotesize{$ \frac{p-1}{2}\left( p^{n-2}-\epsilon \sqrt{p}^{n+s-3}\right) $} & \footnotesize{$p^n-p^{n-s}+\frac{p-1}{2}\left(p^{n-s-1} +\epsilon    \sqrt{p}^{n-s-1} \right)$}\\
\hline
\footnotesize{ $ \frac{1}{2}\left((p-1)p^{n-2}-\epsilon(p+1) \sqrt{p}^{n+s-3}\right)$}& \footnotesize{$\frac{p-1}{2}\left( p^{n-s-1} - \epsilon     \sqrt{p}^{n-s-1}\right)  $}\\
\hline
\end{tabular}\end{center}	
\caption{
\label{table444}  The weight distribution of  $\sC_{\bar{D}_{f,nsq}}$  when $p \equiv 1\pmod 4$ and $n+s$ is odd}
\end{table}

\begin{table}[!htp]
\begin{center}	
\begin{tabular}{|c|c|c|}
\hline
Hamming weight $w$   & Multiplicity $A_w$ \\
\hline
\hline
\footnotesize{$0$} & \footnotesize{$1$}\\
\hline
\footnotesize{$\frac{p-1}{2}p^{n-2}  $} & \footnotesize{$p^{n-s-1}-1   $}\\
\hline
\footnotesize{$ \frac{1}{2}\left((p-1)p^{n-2}+\epsilon(p+1) \sqrt{p^*}^{n+s-3}\right) $} & \footnotesize{$\frac{p-1}{2}\left(p^{n-s-1} +\epsilon    (-1)^n   \sqrt{p^*}^{n-s-1} \right)$}\\
\hline
\footnotesize{ $ \frac{p-1}{2}\left(p^{n-2}+\epsilon  \sqrt{p^*}^{n+s-3}\right)$}& \footnotesize{$p^n-p^{n-s}+\frac{p-1}{2}\left( p^{n-s-1} - \epsilon  (-1)^n \sqrt{p^*}^{n-s-1}\right)  $}\\
\hline
\end{tabular}\end{center}	
\caption{
\label{table444.}  The weight distribution of  $\sC_{\bar{D}_{f,nsq}}$  when $p \equiv 3\pmod 4$ and $n+s$ is odd}

\end{table}

\end{corollary}

 \begin{example} The punctured version $\sC_{\bar{D}_{f,nsq}}$ of $\sC_{D_{f,nsq}}$   in Example  \ref{ExampleNSQ}  is the three-weight   linear  code with parameters $[108,6,63]_3$,  weight enumerator   $1+ 72y^{63} + 576y^{72} +80y^{81}$ and  weight distribution $(1,72,576,80)$. 
This code is minimal by Lemma \ref{Minimality} and  is almost optimal owing to the Griesmer bound. 
\end{example}

With the similar definition of the subcode $\bar{\sC}_{D_{f}}$ defined by (\ref{LinearSubCodeS}), we have 
  a linear code involving $D_{f,sq}$ defined by
\be\nn
\bar{\sC}_{D_{f,sq}}=\{c_\beta= (\Tr^n(\beta d'_1), \Tr^n(\beta d'_2), \ldots, \Tr^n(\beta d'_m)): \beta \in Supp(\Wa f )\},
\ee
which is the subcode  of $\sC_{D_{f,sq}}$  defined by (\ref{LinearCodeSQ}). Hence, the following codes   in  Corollaries \ref{Corollary10}, \ref{Corollary11}, \ref{Corollary12} and \ref{Corollary13} are the subcodes of the codes of  
Theorems \ref{Theorem3}, \ref{Theorem4} and Corollaries \ref{Corollary7}, \ref{Corollary8}, respectively.

\begin{corollary}\label{Corollary10}
The subcode $\bar{\sC}_{D_{f,sq}}$  of the code $\sC_{D_{f,sq}}$ of  Theorem \ref{Theorem3}    is the  two-weight linear code with parameters 
$ [\frac{p-1}{2}(p^{n-1}-\epsilon \eta_0(-1) \sqrt{p^*}^{n+s-2}), n-s ]_p$ whose weight distribution is listed in  Table \ref{table33}. 
\begin{table}[!htp]
\begin{center}	
\begin{tabular}{|c|c|c|}
\hline
Hamming weight $w$   & Multiplicity $A_w$ \\
\hline
\hline
\footnotesize{$0$} & \footnotesize{$1$}\\
\hline
\footnotesize{$  \frac{(p-1)^2}{2}p^{n-2}$} 
& \footnotesize{$ p^{n-s-1} +\frac{p-1}{2}\left(p^{n-s-1}+\epsilon \eta_0^{n+1}(-1)\sqrt{p^*}^{n-s-2}\right)-1 $}\\
\hline
\footnotesize{$(p-1)\left(\frac{p-1}{2}p^{n-2} - \epsilon \eta_0(-1) \sqrt{p^*}^{n+s-2}\right)$}
 & \footnotesize{$\frac{p-1}{2}\left(p^{n-s-1} -\epsilon \eta_0^{n+1}(-1)\sqrt{p^*}^{n-s-2}\right)$}\\
\hline
\end{tabular}\end{center}	
\caption{
\label{table33}  The weight distribution of   $\bar{\sC}_{D_{f,sq}}$ when $n+s$ is even}
\end{table}
\end{corollary}

\begin{corollary}\label{Corollary11}
The subcode $\bar{\sC}_{D_{f,sq}}$  of the code $\sC_{D_{f,sq}}$ of Theorem \ref{Theorem4}    is the three-weight linear code with parameters 
$\left[\frac{p-1}{2}(p^{n-1}+\epsilon \sqrt{p^*}^{n+s-1}),n-s\right]_p$ whose weight distribution is listed in  Table \ref{table55}. 
\begin{table}[!htp]
\begin{center}	
\begin{tabular}{|c|c|c|}
\hline
Hamming weight $w$   & Multiplicity $A_w$ \\
\hline
\hline
\footnotesize{$0$} & \footnotesize{$1$}\\
\hline
\footnotesize{$  \frac{(p-1)^2}{2}p^{n-2}$} & \footnotesize{$p^{n-s-1}-1   $}\\
\hline
\footnotesize{$ \frac{p-1}{2}\left((p-1)p^{n-2}+\epsilon (p^*+1) \sqrt{p^*}^{n+s-3}\right) $} & \footnotesize{$ \frac{p-1}{2}\left( p^{n-s-1} + \epsilon \eta_0^n(-1)    \sqrt{p^*}^{n-s-1}\right) $}\\
\hline
\footnotesize{ $\frac{p-1}{2}\left((p-1)p^{n-2}+\epsilon (p^*-1)\sqrt{p^*}^{n+s-3}\right) $}& \footnotesize{$\frac{p-1}{2}\left( p^{n-s-1} - \epsilon \eta_0^n(-1)    \sqrt{p^*}^{n-s-1}\right)  $}\\
\hline
\end{tabular}\end{center}	
\caption{
\label{table55}  The weight distribution of  $\bar{\sC}_{D_{f,sq}}$  when $n+s$ is odd}
\end{table}
\end{corollary}

\begin{corollary}\label{Corollary12}
The subcode  $\bar{\sC}_{\bar{D}_{f,sq}}$  of the code  $\sC_{\bar{D}_{f,sq}}$ of  Corollary \ref{Corollary7}   is  the two-weight linear code with parameters 
$ [\frac{1}{2}(p^{n-1}-\epsilon \eta_0(-1) \sqrt{p^*}^{n+s-2}), n-s ]_p$ whose weight distribution is listed in  Table \ref{table3333}. 
\begin{table}[!htp]
\begin{center}	
\begin{tabular}{|c|c|c|}
\hline
Hamming weight $w$   & Multiplicity $A_w$ \\
\hline
\hline
\footnotesize{$0$} & \footnotesize{$1$}\\
\hline
\footnotesize{$  \frac{(p-1)}{2}p^{n-2}$} 
& \footnotesize{$ p^{n-s-1} +\frac{p-1}{2}\left(p^{n-s-1}+\epsilon \eta_0^{n+1}(-1)\sqrt{p^*}^{n-s-2}\right)-1 $}\\
\hline
\footnotesize{$ \frac{p-1}{2}p^{n-2}- \epsilon \eta_0(-1) \sqrt{p^*}^{n+s-2} $}
 & \footnotesize{$\frac{p-1}{2}\left(p^{n-s-1} -\epsilon \eta_0^{n+1}(-1)\sqrt{p^*}^{n-s-2}\right)$}\\
\hline
\end{tabular}\end{center}	
\caption{
\label{table3333}  The weight distribution of   $\bar{\sC}_{\bar{D}_{f,sq}}$ when $n+s$ is even}
\end{table}
\end{corollary}

\begin{corollary}\label{Corollary13}
The subcode  $\bar{\sC}_{\bar{D}_{f,sq}}$  of the code  $\sC_{\bar{D}_{f,sq}}$ of Corollary \ref{Corollary8}    is the three-weight linear code with parameters 
$\left[\frac{1}{2}(p^{n-1}+\epsilon \sqrt{p^*}^{n+s-1}),n-s\right]_p$ whose weight distribution is listed in  Table \ref{table5555}.
\begin{table}[!htp]
\begin{center}	
\begin{tabular}{|c|c|c|}
\hline
Hamming weight $w$   & Multiplicity $A_w$ \\
\hline
\hline
\footnotesize{$0$} & \footnotesize{$1$}\\
\hline
\footnotesize{$  \frac{(p-1)}{2}p^{n-2}$} & \footnotesize{$p^{n-s-1}-1   $}\\
\hline
\footnotesize{$ \frac{1}{2}\left((p-1)p^{n-2}+\epsilon (p^*+1) \sqrt{p^*}^{n+s-3}\right) $} & \footnotesize{$ \frac{p-1}{2}\left( p^{n-s-1} + \epsilon \eta_0^n(-1)    \sqrt{p^*}^{n-s-1}\right) $}\\
\hline
\footnotesize{ $\frac{1}{2}\left((p-1)p^{n-2}+\epsilon (p^*-1)\sqrt{p^*}^{n+s-3}\right) $}& \footnotesize{$\frac{p-1}{2}\left( p^{n-s-1} - \epsilon \eta_0^n(-1)    \sqrt{p^*}^{n-s-1}\right)  $}\\
\hline
\end{tabular}\end{center}	
\caption{
\label{table5555}  The weight distribution of  $\bar{\sC}_{\bar{D}_{f,sq}}$  when $n+s$ is odd}
\end{table}
\end{corollary}

With the similar definition of the subcode $\bar{\sC}_{D_{f}}$ defined by (\ref{LinearSubCodeS}), we have 
a linear code involving $D_{f,nsq}$  
\be\nn
\bar{\sC}_{D_{f,nsq}}=\{c_\beta= (\Tr^n(\beta d''_1), \Tr^n(\beta d''_2), \ldots, \Tr^n(\beta d''_m)): \beta \in Supp(\Wa f )\},
\ee
which is the subcode  of $\sC_{D_{f,nsq}}$ defined by  (\ref{LinearCodeNSQ}). 
Hence, the following codes   in  Corollaries \ref{Corollary14}  and \ref{Corollary15}  are the subcodes of the constructed codes in 
Theorem  \ref{Theorem5} and Corollary \ref{Corollary9}, respectively.

\begin{corollary}\label{Corollary14}
The subcode  $\bar{\sC}_{D_{f,nsq}}$   of the code  $\sC_{D_{f,nsq}}$  of Theorem \ref{Theorem5}  is the three-weight linear code with parameters 
$\left[\frac{p-1}{2}(p^{n-1}-\epsilon \sqrt{p^*}^{n+s-1}),n-s\right]_p$ whose weight distribution is listed in  Table \ref{table44}. 
\begin{table}[!htp]
\begin{center}	
\begin{tabular}{|c|c|c|}
\hline
Hamming weight $w$   & Multiplicity $A_w$ \\
\hline
\hline
\footnotesize{$0$} & \footnotesize{$1$}\\
\hline
\footnotesize{$\frac{p-1}{2}(p-1)p^{n-2}  $} & \footnotesize{$p^{n-s-1}-1   $}\\
\hline
\footnotesize{$ \frac{p-1}{2}\left((p-1)p^{n-2}-\epsilon(p^*-1) \sqrt{p^*}^{n+s-3}\right) $} & \footnotesize{$\frac{p-1}{2}\left(p^{n-s-1} +\epsilon   \eta_0^n(-1)    \sqrt{p^*}^{n-s-1} \right)$}\\
\hline
\footnotesize{ $ \frac{p-1}{2}\left((p-1)p^{n-2}-\epsilon(p^*+1) \sqrt{p^*}^{n+s-3}\right)$}& \footnotesize{$\frac{p-1}{2}\left( p^{n-s-1} - \epsilon \eta_0^n(-1)    \sqrt{p^*}^{n-s-1}\right)  $}\\
\hline
\end{tabular}\end{center}	
\caption{
\label{table44}  The weight distribution of  $\bar{\sC}_{D_{f,nsq}}$  when $n+s$ is odd}
\end{table}
\end{corollary}
\begin{corollary}\label{Corollary15}
The subcode   $\bar{\sC}_{\bar{D}_{f,nsq}}$  of the code  $\sC_{\bar{D}_{f,nsq}}$ of  Corollary \ref{Corollary9}  is the   three-weight linear code with parameters 
$\left[\frac{1}{2}(p^{n-1}-\epsilon \sqrt{p^*}^{n+s-1}),n-s\right]_p$ whose weight distribution is listed in  Table \ref{table4444}. 
\begin{table}[!htp]
\begin{center}	
\begin{tabular}{|c|c|c|}
\hline
Hamming weight $w$   & Multiplicity $A_w$ \\
\hline
\hline
\footnotesize{$0$} & \footnotesize{$1$}\\
\hline
\footnotesize{$\frac{1}{2}(p-1)p^{n-2}  $} & \footnotesize{$p^{n-s-1}-1   $}\\
\hline
\footnotesize{$ \frac{1}{2}\left((p-1)p^{n-2}-\epsilon(p^*-1) \sqrt{p^*}^{n+s-3}\right) $} & \footnotesize{$\frac{p-1}{2}\left(p^{n-s-1} +\epsilon   \eta_0^n(-1)    \sqrt{p^*}^{n-s-1} \right)$}\\
\hline
\footnotesize{ $ \frac{1}{2}\left((p-1)p^{n-2}-\epsilon(p^*+1) \sqrt{p^*}^{n+s-3}\right)$}& \footnotesize{$\frac{p-1}{2}\left( p^{n-s-1} - \epsilon \eta_0^n(-1)    \sqrt{p^*}^{n-s-1}\right)  $}\\
\hline
\end{tabular}\end{center}	
\caption{
\label{table4444}  The weight distribution of  $\bar{\sC}_{\bar{D}_{f,nsq}}$  when $n+s$ is odd}
\end{table}
\end{corollary}
  \begin{remark}
When we assume only the weakly regular bent-ness in this subsection, we can obviously recover the linear codes obtained by Tang et al.   \cite{tang2016linear}. Therefore, this subsection can be viewed as an extension of \cite{tang2016linear} to the notion of weakly regular  plateaued functions.
\end{remark}
The following natural question may now spring to mind: Are the constructed codes in this section  minimal? The following section investigates the minimality of the constructed codes.

\section{The minimality of the constructed  linear codes}\label{SectionSSS}
This section confirms that the constructed codes   from weakly regular    plateaued functions  in Section \ref{Constructions}  are minimal. In other words, with the help of Lemma \ref{Minimality}, we   observe  that  all   nonzero codewords of the constructed  codes are minimal  for almost all cases. 
To do this, we consider separately   the constructed codes  in
Theorems \ref{Theorem1},  \ref{Theorem2},  \ref{Theorem3},  \ref{Theorem4} and   \ref{Theorem5}.

\begin{theorem}\label{Theorem6}
Let $n+s$ be an even integer.  If $\epsilon \eta_0^{(n+s)/2}(-1)=1$, then the linear code $\sC_{D_f}$ of Theorem \ref{Theorem1} is  minimal       with parameters $$\left[ p^{n-1}-1+  (p-1) p^{(n+s-2)/2},n,(p-1)p^{n-2}\right]_p$$ when  $0\leq s\leq n-4$; otherwise,  $\left[ p^{n-1}-1-  (p-1) p^{(n+s-2)/2},n,(p-1)\left(p^{n-2}-p^{(n+s-2)/2}\right)\right]_p$ when $0\leq s\leq n-6$.
\end{theorem}

\begin{proof} If $\epsilon \eta_0^{(n+s)/2}(-1)=1$, then 
$w_{\min}=(p-1)p^{n-2}\mbox{ and } 
w_{\max}=(p-1)\left(p^{n-2} + p^{(n+s-2)/2}\right)$; otherwise,  
$w_{\min}=(p-1)\left(p^{n-2}- p^{(n+s-2)/2}\right) \mbox{ and } 
w_{\max}=(p-1)p^{n-2}.$ 
In the  first case, 
we observe that
$$ \frac{p-1}{p} <\frac{(p-1)p^{n-2}}{(p-1)\left(p^{n-2} + p^{(n+s-2)/2}\right)}$$
if   $0\leq s \leq n-4$.
Similarly, in the  second case, 
we observe that
$$ \frac{p-1}{p} <\frac{(p-1)\left(p^{n-2}- p^{(n+s-2)/2}\right)}{(p-1)p^{n-2}}$$
if $0\leq s \leq n-6$.
Hence, the proof is  completed from Lemma \ref{Minimality}.
\end{proof} 

 \begin{corollary}
The constructed codes in Corollaries \ref{Corollary1}, \ref{Corollary3} and \ref{Corollary5}  are minimal  with the corresponding condition in Theorem \ref{Theorem6}.
\end{corollary}
 
\begin{theorem}\label{Theorem7}
Let $n+s$ be an odd integer with $0\leq s\leq n-5$.  Then the linear code $\sC_{D_f}$ 
of Theorem \ref{Theorem2} is   minimal     with parameters
 $\left[ p^{n-1}-1,n,(p-1)\left(p^{n-2}- p^{(n+s-3)/2}\right)\right]_p$.
\end{theorem}

\begin{proof}  There are two cases:  $\epsilon \eta_0^{(n+s-3)/2}(-1)=\pm 1$. For both cases,  we have  
$w_{\min}=(p-1)\left(p^{n-2}- p^{(n+s-3)/2}\right)$ and
$w_{\max}=(p-1)\left(p^{n-2}+  p^{(n+s-3)/2}\right).$
Then  we observe that 
$$\frac{p-1}{p} <\frac{w_{\min}}{w_{\max}}$$
if $0\leq s\leq n-5$. It then follows from Lemma \ref{Minimality} that all nonzero codewords of $\sC_{D_f}$ are minimal if $0\leq s\leq n-5$.
\end{proof} 

\begin{corollary} Let $n+s$ be an odd integer with $0\leq s\leq n-5$. 
Then the constructed codes in Corollaries \ref{Corollary2}, \ref{Corollary4} and \ref{Corollary6}  are minimal. 
\end{corollary}

\begin{theorem}\label{Theorem8}
Let $n+s$ be an even integer.  If $\epsilon \eta_0^{(n+s)/2}(-1)=1$, then the linear code $\sC_{D_{f,sq}}$ of  Theorem \ref{Theorem3} is  minimal     with parameters
\be\nn
\left[ \frac{p-1}{2}\left(p^{n-1}-  p^{(n+s-2)/2}\right), n, \frac{(p-1)^2}{2}p^{n-2}-  (p-1) p^{(n+s-2)/2}\right]_p
\ee
when $0\leq s\leq n-6$; otherwise, $\left[ \frac{p-1}{2}\left(p^{n-1}+p^{(n+s-2)/2}\right), n,\frac{(p-1)^2}{2}p^{n-2} \right]_p$
when $0\leq s\leq n-4$.
\end{theorem}

\begin{proof} If $\epsilon \eta_0^{(n+s)/2}(-1)=1$, then
$w_{\min}=\frac{(p-1)^2}{2}p^{n-2}-  (p-1) p^{(n+s-2)/2} \mbox{ and }
w_{\max}=\frac{(p-1)^2}{2}p^{n-2}$; otherwise, 
$w_{\min}=\frac{(p-1)^2}{2}p^{n-2} \mbox{ and }
w_{\max}=\frac{(p-1)^2}{2}p^{n-2}+ (p-1) p^{(n+s-2)/2}.$
In the  first case, we have 
$$ \frac{p-1}{p} <\frac{\frac{(p-1)^2}{2}p^{n-2}-  (p-1) p^{(n+s-2)/2}}{\frac{(p-1)^2}{2}p^{n-2}}$$
if $0\leq s\leq n-6$, and in the second case, we have
$$ \frac{p-1}{p} <\frac{\frac{(p-1)^2}{2}p^{n-2} }{\frac{(p-1)^2}{2}p^{n-2}+ (p-1) p^{(n+s-2)/2}}$$ if $0\leq s\leq n-4$.
Hence, the proof is completed by Lemma \ref{Minimality}.
\end{proof} 
 \begin{corollary}
The constructed codes in Corollaries \ref{Corollary7}, \ref{Corollary10} and \ref{Corollary12} are minimal  with the corresponding condition in Theorem \ref{Theorem8}.
\end{corollary}

\begin{theorem}\label{Theorem9}
Let $n+s$ be an odd integer with $0\leq s\leq n-5$.  Then 
the linear code  $\sC_{D_{f,sq}}$ of Theorem \ref{Theorem4} is   minimal     with parameters

\be\nn
\begin{array}{ll}
 \left[\frac{p-1}{2}(p^{n-1}+ p^{(n+s-1)/2}),n, \frac{(p-1)^2}{2}p^{n-2}\right]_p,&\mbox{if } \epsilon \eta_0^{(n+s-1)/2}(-1)=1,\\
 \left[\frac{p-1}{2}(p^{n-1}- p^{(n+s-1)/2}),n, \frac{p-1}{2}\left((p-1)p^{n-2}- (p+1) p^{(n+s-3)/2}\right)\right]_p, & \mbox{ otherwise.}
\end{array}
\ee  
\end{theorem}

\begin{proof} When $p\equiv 1 \pmod 4$, 
we have 
\be\label{minweight}
\begin{array}{ll}
w_{\min}&=\frac{(p-1)^2}{2}p^{n-2} \mbox{ and }\\
w_{\max}&=  \frac{p-1}{2}\left((p-1)p^{n-2}+(p+1) p^{(n+s-3)/2}\right)
\end{array}
\ee
if  $\epsilon=1$; otherwise, we have
\be\label{maxweight}
\begin{array}{ll}
w_{\min}&=\frac{p-1}{2}\left((p-1)p^{n-2}- (p+1) p^{(n+s-3)/2}\right) \mbox{ and } \\
w_{\max}&=\frac{(p-1)^2}{2}p^{n-2}.
\end{array}
\ee
When $p\equiv 3 \pmod 4$,   we have the Hamming weights  in (\ref{minweight}) if $\epsilon\eta^{(n+s-1)/2}_0(-1)=1$; otherwise,  in (\ref{maxweight}).
For each case above,  we have $$\frac{p-1}{p} <\frac{w_{\min}}{w_{\max}}$$
if  $0\leq s\leq n-5$.  Hence, Lemma \ref{Minimality} completes the proof.
\end{proof}

 \begin{corollary} Let $n+s$ be an odd integer with $0\leq s\leq n-5$. Then
the constructed codes in Corollaries \ref{Corollary8}, \ref{Corollary11} and \ref{Corollary13} are minimal. 
\end{corollary}

\begin{theorem}\label{Theorem10}
Let $n+s$ be an odd integer with $0\leq s\leq n-5$.  Then the linear code
$\sC_{D_{f,nsq}}$  of Theorem \ref{Theorem5} is   minimal     with parameters
\be\nn
\begin{array}{ll}
\left[\frac{p-1}{2}(p^{n-1}- p^{(n+s-1)/2}),n,\frac{p-1}{2}\left((p-1)p^{n-2}- (p+1) p^{(n+s-3)/2}\right)\right]_p,&\mbox{if } \epsilon \eta_0^{(n+s-1)/2}(-1)=1,\\
 \left[\frac{p-1}{2}(p^{n-1}+  p^{(n+s-1)/2}),n,\frac{(p-1)^2}{2}p^{n-2}\right]_p, & \mbox{ otherwise.}
\end{array}
\ee  
\end{theorem}

\begin{proof}  When $p\equiv 1 \pmod 4$, we have   the Hamming weights in (\ref{maxweight})  if  $\epsilon=1$; otherwise, in (\ref{minweight}).
Similarly, in the case of $p\equiv 3 \pmod 4$,   we have the Hamming weights in    (\ref{maxweight}) if  $\epsilon\eta^{(n+s-1)/2}_0(-1)=1$; otherwise, in (\ref{minweight}).
 Hence,  the assertion follows directly from Theorem \ref{Theorem9}. 
\end{proof}
  \begin{corollary} Let $n+s$ be an odd integer with $0\leq s\leq n-5$. 
Then the constructed codes in Corollaries \ref{Corollary9}, \ref{Corollary14} and \ref{Corollary15}  are minimal. 
\end{corollary}

\begin{remark}
We conclude from this section that the constructed codes in this paper  are minimal for almost all cases. Hence,  the secret sharing schemes based on the dual codes of the constructed minimal linear codes in this paper have the nice access structures described in Proposition \ref{Structure}.  This is the motivation why we construct  a punctured version and a subcode of each constructed code. 
\end{remark}

\section{Conclusion}
In this paper, inspired by the work of \cite{tang2016linear},  we push further the use of weakly regular plateaued functions over finite fields of odd characteristic  introduced  recently by Mesnager et al.    \cite{mesnager2017WCC}.  By generalizing the linear codes constructed from weakly regular bent functions in \cite{tang2016linear}, we obtain new minimal linear codes with more freedom in the choice of the functions involved in the construction of two or three weight linear codes. They contain  the (almost) optimal codes with respect to the Singleton and Griesmer bounds.
 The paper provides the first construction of linear codes with few weights from weakly regular plateaued functions based on the second generic construction.
 The obtained  minimal codes in this paper can be directly used to construct secret sharing schemes with the nice access structures.
To the best of our knowledge, they are  inequivalent to the known ones (since there is no minimal linear code with these parameters) in the literature.

\bibliography{myBiblio}
\bibliographystyle{abbrv}
\bibliographystyle{iamBiblioStyle}

\end{document}